\documentclass[journal]{IEEEtran}
\usepackage{amsmath,amsfonts}
\usepackage{algorithmic}
\usepackage{algorithm}
\usepackage{array}
\usepackage[caption=false,font=footnotesize,labelfont=rm,textfont=rm]{subfig}
\usepackage{textcomp}
\usepackage{stfloats}
\usepackage{url}
\usepackage{verbatim}
\usepackage{graphicx}
\usepackage{cite}
\usepackage{color}
\usepackage{amssymb}
\usepackage{amsthm}
\usepackage{bm}
\usepackage{amsmath} % 加载 amsmath 包

\newtheorem{theorem}{Theorem}
\newtheorem{Lemma}{Lemma}
\newtheorem{corollary}{Corollary}

\newtheorem{definition}{Definition}
\newtheorem{remark}{Remark} % 定义 remark 环境

\def\BibTeX{{\rm B\kern-.05em{\sc i\kern-.025em b}\kern-.08em
    T\kern-.1667em\lower.7ex\hbox{E}\kern-.125emX}}
\begin{document}

\title{Multi-Source Peak Age of Information Optimization in Mobile Edge Computing Systems}
 \author{Jianhang~Zhu and Jie~Gong,~\IEEEmembership{Member,~IEEE}
 \IEEEcompsocitemizethanks
 {\IEEEcompsocthanksitem J. Zhu and J. Gong are with the School of Computer Science and Engineering, and the Guangdong Key Laboratory of Information Security Technology, Sun Yat-sen University, Guangzhou 510006, China. Emails: zhujh26@mail2.sysu.edu.cn, gongj26@mail.sysu.edu.cn.}
 }

% \author{
%     \IEEEauthorblockN{
%         Jianhang Zhu and Jie Gong
%     }
%     \IEEEauthorblockA{
%         \textit{Guangdong Key Laboratory of Information Security Technology},\\
%         \textit{School of Computer Science and Engineering},
%         \textit{Sun Yat-sen University},
%         Guangzhou, 510006, China \\
%         Email:
%         zhujh26@mail2.sysu.edu.cn,
%         gongj26@mail.sysu.edu.cn
%     }
% }

\maketitle
\begin{abstract}
  Age of Information (AoI) is emerging as a novel metric for measuring information freshness in real-time monitoring systems. For computation-intensive status data, the information is not revealed until being processed. We consider a status update problem in a multi-source single-server system where the sources are scheduled to generate and transmit status data which are received and processed at the edge server. Generate-at-will sources with both random transmission time and process time are considered, introducing the joint optimization of source scheduling and status sampling on the basis of transmission-computation balancing. We show that a random scheduler is optimal for both non-preemptive and preemptive server settings, and the optimal sampler depends on the scheduling result and its structure remains consistent with the single-source system, i.e., threshold-based sampler for non-preemptive case and transmission-aware deterministic sampler for preemptive case. { Then, the problem can be transformed to jointly optimizing the scheduling frequencies and the sampling thresholds/functions, which is non-convex. We proposed an alternation optimization algorithm to solve it. Numerical experiments show that the proposed algorithm can achieve the optimal in a wide range of settings.}
\end{abstract}

\begin{IEEEkeywords}
Age of information, multi-source, transmission-computation trade-off, mobile edge computing, service preemption
\end{IEEEkeywords}

\section{Introduction}

Recently, shifting cloud functions to network edges such as base stations and access points, is a new trend in mobile computing to utilize the vast amount of idle computation power and storage space for computation-intensive and latency-critical tasks of mobile devices. This trend is known as Mobile Edge Computing (MEC) \cite{mao2017survey}. It provides low-latency services, and combined with the extensive update collection capabilities of the Internet of Things (IoT) \cite{IoT2015}, it enables various real-time applications such as remote monitoring and control, phase packet update in smart grids, and environment monitoring for autonomous driving. The performance of these systems is closely tied to the information freshness. The concept of \emph{Age of Information} (AoI), defined as the time elapsed since the generation of the last received update, was introduced in \cite{kaul2012real} to quantify information freshness. {For status monitoring applications with computation-intensive updates such as image detection, the status information embedded in an update packet cannot be exposed until it is being processed. Hence, it is essential to jointly optimize transmission and computing for AoI minimization.
}

The emergence of AoI has inspired numerous studies that combine queuing theory to analyze and optimize AoI under different queue and server settings\cite{kaul2012real,yate2012real,sun2017update,arafa2020age}. Additionally, optimizing AoI in wireless networks has gained increasing attention in multi-source networks \cite{yates2019theage, yate2012real, huang2015opt, Najm2018status, Moltafet2020onthe, bedewy2021optimal, talak2020opt}, where scheduling strategies need to be designed to serve different sources and optimize the system average AoI. In multi-source MEC systems, both multi-source scheduling and the impact of the computation process needs to be considered, thus necessitating the joint consideration of scheduling and sampling to reduce the system average AoI.

To address this issue, we consider a multi-source status update system with MEC as shown in Fig. \ref{fig:system}. There are \(M\) sources monitoring different time-varying processes. At any time, the sources sample the time-varying processes and generate updates, which are transmitted through a delayed channel to the edge server. The server then processes these updates and sends the results to the destination. The server has a queue to cache the received updates when necessary. Assume transmission time and computation time are both random and unknown a priori. A scheduler-sampler at the edge server decides which source should sample a new data at which time. {The transmit-then-compute process of edge computing can be viewed as a two-hop system. To optimize AoI in MEC systems, it is essential to understand how the actions of two hops interact with each other. It also sets up a foundation for the possible extension to multi-hop systems such as partial offloading in edge computing.
}

We adopt the weighted sum average peak AoI (PAoI) as the performance metric. The sum weights indicate the different levels of interest of the destination to different sources. {PAoI is a meaningful data freshness metric for scenarios that focus on reducing AoI voilation probability \cite{costa2014age, costa2016PAoI}. In these applications, an update is viewed as fresh if its age is within a certain threshold. We focus on the average PAoI as it is important to characterize PAoI.} In addition, the average PAoI is usually more analytically tractable than average AoI and leads to well-structured solutions \cite{cham2021min}. Especially for MEC systems, it is quite challenging to optimize the average AoI, and only heuristic strategies were proposed~\cite{zhu2022online}. Hence, we focus on average PAoI optimization problem and aim to derive well-structured optimal solution.

\begin{figure}[!t]
  \centering
  \includegraphics[width=0.7\linewidth]{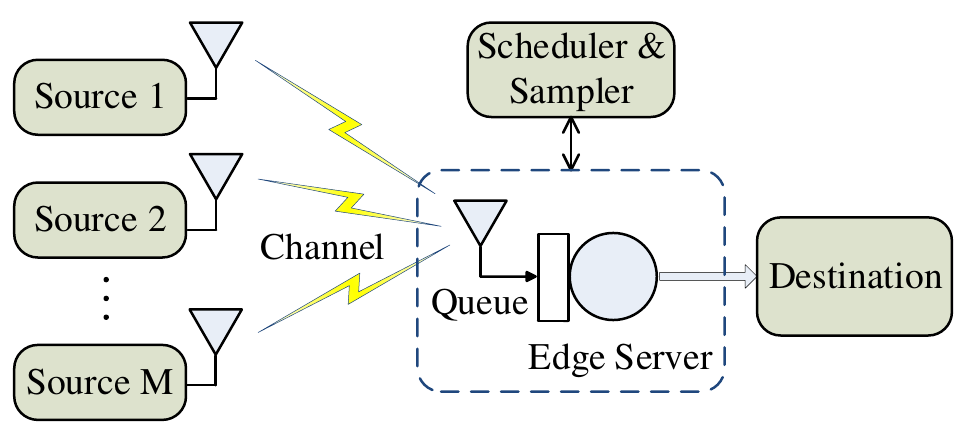}
  \caption{Multi-source status update system with MEC. A scheduler-sampler manages the scheduling of updates from $M$ sources. Updates from different sources are transmitted through a shared channel to the edge server for processing.} \label{fig:system}
\end{figure}

In a single-source MEC system, balancing transmission and computation time is key to optimizing system real-time performance. Our previous work~\cite{jian2024opt} considered the joint optimization of transmission and computation in single-source case. It revealed the structure of optimal samplers for both preemptive and non-preemptive server settings. {However, in multi-source systems, besides balancing transmission and computation processes, scheduling among sources need to be further considered. In particular, updates from one source may delay the processing of others, and thus the AoI performance of one source is affected by the updates of others. This interdependence significantly complicates the AoI analysis and optimization. Therefore, the source scheduling and sampling are coupled among multiple sources, which is the key challenge in multi-source case that distinguishes from the single-source case. This work focuses on dealing with this challenge.}

\subsection{Main Results}

In this paper, we consider the problem of joint transmission and computation optimization in a multi-source system. We aim at minimizing the multi-source weighted sum average PAoI by optimizing the scheduling-sampling policy. We assume that both the transmission time and the computation time are independently and identically distributed (i.i.d.). Sources can sample at any time, referred to as the \emph{generate-at-will} model. We derive a mathematical expression for the weighted sum average PAoI, applicable to both preemptive and non-preemptive edge server settings. The feasible strategy set includes all causal schedulers and samplers, where control decisions depend on the system's history and current information. Our main contributions are as follows:
\begin{itemize}
  \item We prove that the performance of the optimal scheduler only depends on scheduling frequency vector in both non-preemptive and preemptive systems (Theorem~\ref{The:WOP}). Based on this, we define a random scheduler (Definition~\ref{def:RS}), which is optimal in both systems (Corollary~\ref{coro1}). Under the random scheduler, we demonstrate that a continuous working sampler is optimal for both systems (Lemma \ref{lem:maxZ}). We further derive the expressions of the optimal scheduling frequencies for both system settings (Lemma~\ref{lem:optimalf} and Lemma~\ref{lem:optimalf_p}) for given optimal smaplers. 
  \item In a non-preemptive system, we show that the optimal sampler depends on a set of fixed thresholds (Theorem~\ref{The:WOP_FMT}). The problem of computing the optimal thresholds can be decomposed into independent single-threshold optimization problems, each can be solved as a single-source problem. Then, we propose an alternating optimization algorithm to optimize the scheduling frequencies and sampling thresholds (Algorithm~\ref{alg:optimal_policy}).
  \item In a preemptive system, we show that the optimal sampler depends on a function of the data transmission time and the scheduling index (Theorem~\ref{The:WP_SD}). The sampler function can be determined by iteratively optimizing each source-dependent sampler which can be calculated via the Dinkelbach algorithm (Algorithm~\ref{alg:Dinkelbach}).
  \item Numerical experiments validate the analysis results for both systems. We also examined the variation of the optimal sampler parameters under the random scheduler. The results show that when the average transmission time is large, the optimal sampler always degenerates into a zero-wait sampler.
\end{itemize}

\subsection{Related Work}

\textbf{Source Settings:} The Age of Information has become an important metric in various status update systems. Related studies can be classified into two categories based on the type of source. The first category focuses on uncontrollable sources, as seen in works such as \cite{kaul2012real,yate2012real,bt2015age,huang2015opt,bed2019the,yates2019theage}. In these works, minimizing AoI typically involves optimizing the service rate and packet management strategy in the queues. For instance, ref.~\cite{kaul2012real} analyzed the average AoI in the First-Come-First-Serve (FCFS) system for different types of queues and found that there exists an optimal offered load for each queue to minimize the average AoI.  In \cite{yates2019theage}, a new simplified technique for evaluating AoI in finite-state continuous-time queuing systems was derived for the multi-source LCFS system with Poisson arrivals and exponential service time. The second category considers controllable sources, also known as the generate-at-will source model, as seen in \cite{sun2017update,arafa2020age,gu2021opt,cham2021min,gong2022sleep,arafa2021timely}. For example, ref.~\cite{sun2017update} considered a generate-at-will source model where updates can be generated at any time according to a scheduling policy and proved that a threshold-based policy is optimal when considering independent and identically distributed transmission time. However, these works did not consider the impact of computation on AoI.

\textbf{Computing Time:} The impact of computing on AoI has garnered increasing attention~\cite{alabbasi2018joint, arafa2019timely, zou2021opt, kuang2019age, gong2019reducing, zhong2019age, song2019age, li2021age, zhu2022online}, as the information contained in a status update packet is often not revealed until it has been processed. Ref.~\cite{alabbasi2018joint} explored the impact of computing on AoI by scheduling computing tasks in the central cloud. The scheduling policy for cloud computing ignoring transmission time was studied in~\cite{arafa2019timely}. In~\cite{zou2021opt}, the trade-off between computing and transmission was analyzed, where each packet was pre-processed before being transmitted. In~\cite{kuang2019age} and~\cite{gong2019reducing}, the average AoI with exponential transmission time and service time was analyzed for the single-user case when MEC was considered. However, these works did not address the fundamental question of what the optimal scheduling policy is to achieve the minimum AoI in an MEC system that considers both transmission and computation.

\textbf{Peak Age of Information:} PAoI is a variant of AoI that reflects the average of the maximum AoI values, introduced in \cite{costa2014age}. In recent years, it has been further investigated in the literature \cite{abd2019aPAoI, ino2019agen, Bedewy2021Low-Power,Yang2022Edge}. In \cite{abd2019aPAoI}, the role of unmanned aerial vehicles (UAV) as mobile relays to minimize the average PAoI for a source-destination pair was explored. An optimization problem was formulated to jointly optimize the UAV's flight trajectory as well as energy and service time allocations for packet transmissions. Although PAoI is not an age penalty functional~\cite{yates2021age}, it is still effective for AoI optimization. Ref.~\cite{ino2019agen} indicated that the stationary distribution of AoI can be expressed in terms of the stationary distributions of the system delay and PAoI, with the average PAoI serving as a fundamental characterization of the age process.

\textbf{Multi-Source:} Many works have considered the analysis and optimization of AoI in multi-source single-server systems \cite{yates2019theage, yate2012real, huang2015opt, Najm2018status, Moltafet2020onthe, bedewy2021optimal}. Some of these works assume that sources generate updates according to independent Poisson processes, such as in M/G/1/1 queues \cite{Najm2018status}. In \cite{bedewy2021optimal, talak2020opt}, the authors considered the multi-source scheduling-sampling optimization problem, but they did not take into account the impact of computation time, which is modeled and studied in our work. Our previous work~\cite{jian2024opt} derived structural optimal sampling policies for a single source system, which is shown later in this paper to also hold for multi-source case. In addition to that, this work further focused on the challenging multi-source scheduling problem and how to jointly optimize source scheduling and status sampling.

\begin{figure*}[!t]
  \begin{minipage}{1\linewidth}
      \centering
      \subfloat[Age of information in the non-preemptive system]{\includegraphics[width=.4\linewidth]{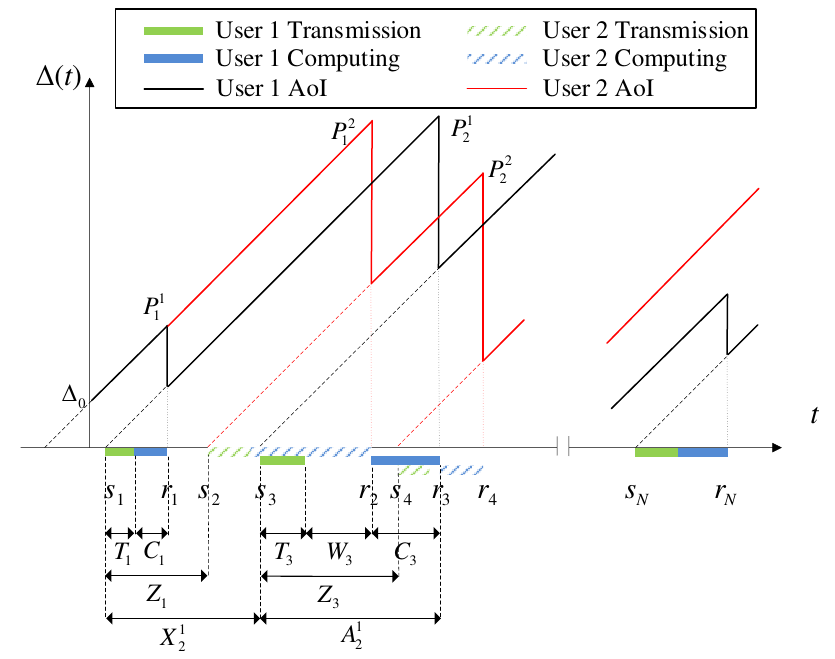}  \label{fig:np}}
      \subfloat[Age of information in the  preemptive system]{\includegraphics[width=.4\linewidth]{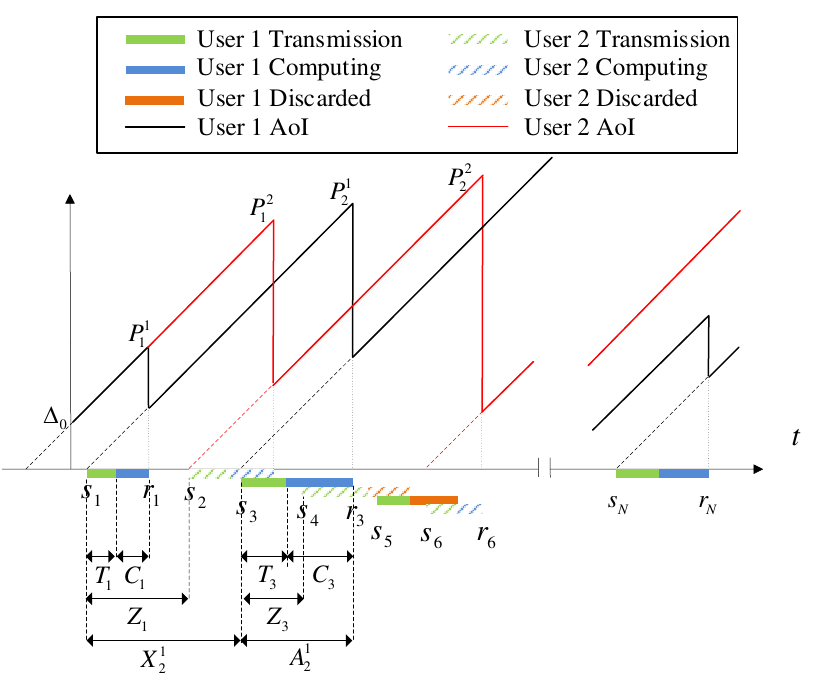} \label{fig:p}} 
  \end{minipage}
  \caption{Age curves in non-preemptive system (a) and preemptive system (b). Both systems consist of two users, and the initial AoI for each user is $\Delta_0$. Green and blue rectangles represent the transmission and computation processes of the data, respectively, while orange rectangles represent preempted data (in the preemptive system). Solid rectangles and dashed rectangles are used to distinguish between the two users.}
  \label{fig:Age_curve}
\end{figure*}

\section{System Model And Problem Formulation}\label{sec:sys}
\subsection{System Model}\label{subsec:systemmodel}

We consider a status update MEC system with $M$ sources as shown in Fig.~\ref{fig:system}, where each source observes a time-varying process and generates real-time status updates. A status update generated by a source is sent to an edge server though a delay channel. A scheduler-sampler at the edge server controls the transmission order of sources and the generation time of update from each source. Once the control decision is made, the scheduler-sampler informs the scheduled source at the intended sampling time to generate and transmit a new update via a dedicated control channel. This is known as the generate-at-will model, meaning that updates is generated under control.

Each update packet arrives at the edge server after a random transmission time, and then is served with a random computation time, and finally is delivered to the destination. { We assume the channel is able to transmit only one packet at a time, and the edger server can process only one packet at a time as well. This fits to the time division systems in practice. In this case, packet queue may occur when either the channel or the server is busy. Since an update becomes stale while waiting in queues, it is essential to mitigate queuing delay. The transmission queuing delay can be completely eliminated by generating a new update only after the channel is idle. But in this case, the updates may accumulate at the server queue causing long computing queuing delay. To completely eliminate the computing queuing delay, updates should be generated after both the channel and the server are idle, which however, results in high AoI due to long update interval. Intuitively, AoI can be reduced by transmitting a new packet when the previous one is under processing to speed up update. To take advantage of this feature, we assume that a new update can be generated only after the server queue is empty. Although the computing queuing delay still exists, it only affects the current update without accumulative effect, and can be effectively reduced by careful policy design. With this strategy, there is at most one packet in the server queue. Hence, we consider a unit-sized server queue. In practice, our proposed methods can work with any non-zero queue size.
}

% Updates which arrive when the edge server is busy can be stored in a queue.  The decision maker is aware of the idle/busy state of the channel and the edge server through the feedback link between the scheduler and the edge server. We assume the transmission time of the feedback signal is negligible.

We optimize the information age under two edge server settings:
\begin{enumerate}
  \item \textit{System without service preemption.} In the non-preemptive service case, newly arrived updates are stored in the unit-sized queue if the edge server is busy, and the updates wait until the server becomes idle. During the waiting period, none of the sources generate new updates. Thus, all updates are successfully delivered.
  \item \textit{System with service preemption.} In the preemptive service case, the edge server always serves the newly arrived packet and discards the old one. Therefore, the queue is always empty. No updates need to wait in the queue, but some updates will be dropped due to preemption.
\end{enumerate}

Suppose the update $i$ is generated and submitted at time $s_i$, its transmission time is $T_i$, and its computation time is $C_i$. We use $m_i$ to represent the source index from which packet $i$ is generated. Assume the transmission time of the updates of all sources are i.i.d with a positive finite mean $0<E[T]<\infty$, and the computation time are also i.i.d. with mean $0 < E[C] < \infty$. {It refers to the scenarios with the same transmission environment where the average transmission time can be set the same by power control and the same monitoring task which is processed via the same execution program such as a neural network with the same computing hardware. The extension to non-i.i.d.~case is briefly discussed later and left for future work.} Let $F_{T}(\cdot),f_{T}(\cdot)$ and $F_{C}(\cdot),f_{C}(\cdot)$ denote the cumulative distribution function and probability density function of the transmission time and computation time, respectively. Since both $T$ and $C$ are positive, we have $f_T(t) = f_C(t) = 0 $ for all $t \le 0$. Denote $W_i$ as the waiting time of packet $i$ in the queue. Hence, the packet $i$ is received by the destination at time
\begin{align}
  r_i=s_i+T_i+W_i+C_i\label{equ:dk}
\end{align}
or is possibly dropped (in preemptive service case). Note that it is possible that $W_i = 0$, which is then omitted in Fig.~\ref{fig:Age_curve}. We set $r_i=\infty$ if packet $i$ is dropped due to preemption. Assume that packet $0$ is submitted at time $s_0 =-T_0-C_0$ and is delivered at $r_0=0$. The inter-generation time between packet $i+1$ and packet $i$, denoted by $Z_i$, is given by 
\begin{align}
  Z_i=s_{i+1}-s_i.\label{equ:zk}
\end{align}
Since the source can only submit updates when the channel is idle and the queue is empty, we have 
\begin{align}
  Z_i\ge T_i+W_i.
\end{align}
A scheduling-sampling strategy, denoted as \(d\), consists of two components: 1) a scheduler, denoted as \(\xi\), determining the sources scheduled for each transmission moment \(\xi = (m_1, m_2, ...)\); 2) a causal sampler, denoted as \(\pi\), determining the instants when new update updates are generated \(\pi = (s_1, s_2, ...)\), which can be equivalently expressed as \(\pi = (Z_1, Z_2, ...)\).

At time $t$, the AoI of source $m$ at the destination, denoted by $\Delta_m(t)$, is given by
\begin{equation}
\Delta_m(t)=t-\mathop{\max}_{i\in \mathbb{N}}\{s_i|m_i=m,r_i\le t\}.\label{def:age}
\end{equation}
The initial AoI is $\Delta_m(0)=T_0+C_0$ for all sources. Fig. \eqref{fig:Age_curve} provides an example of both non-preemptive and preemptive systems involving two users.

\subsection{Problem Formulation}

Consider source $m$ has a weight \(w_m\), representing the level of demand from the destination for update from different sources. For simplicity, we assume \(\sum_{m=1}^M w_m = 1\). Under a given scheduling policy $d=(\xi,\pi)$, the weighted sum average PAoI is defined as
\begin{align}
  \bar{P}_{d}= \lim_{{K \to \infty}}  \sum_{m=1}^{M}w_m \mathbb{E}\left[ \frac{\sum_{i=1}^{K}\Delta_m(r_i^{-})\mathbf{1}_{m_i=m} }{\sum_{i=1}^{K}\mathbf{1}_{m_i=m} }\right],\label{def:PAoI}
\end{align}
where $\mathbf{1}_{E}$ is the indicator function of event $E$, $r_i^{-}$ is the left limit to the delivery time of packet $i$, and $K$ is the total number of updates received at the destination. Since we focus on a set of policies where the source can only submit packets when the channel is idle and the queue is empty, we define the set of all feasible policies as $\mathcal D$, which includes causal schedulers and causal samplers $\Pi=\{\pi={Z_i,i\ge 1}|Z_i\ge T_i+W_i\}$. The problem of minimizing the weighted average PAoI can be formulated as:
\begin{align}
  \bar{P}^{*,\omega}=\min_{d \in \mathcal D}\bar{P}_{d}\label{def:prob},
\end{align}
where the superscript $\omega\equiv \text{wop}$ is the case without preemption and $\omega\equiv \text{wp}$ is the case with preemption. Our goal is to find the optimal policy in $\mathcal D$, denoted by $d^*=(\xi^*,\pi^*)$, that achieves the minimum average PAoI $\bar{P}^{*,\omega}$. When $M=1$, there is only one feasible scheduler. Our previous work~\cite{jian2024opt} considers the sampling problem \eqref{def:prob} in this case and provides the structure of the optimal sampler under two different server settings. In the following, we consider the optimal policy in multi-source systems.

\subsection{Symbol Definitions for Each Source}

For both systems, we denote: (i) $i_k^m$ as the system packet index of the $k$-th generated update packet of source $m$; (ii) the inter-generation time between two consecutive updates of source $m$ as  
\begin{align}  
  X_k^m = s_{i_{k}^{m}} - s_{i_{k-1}^m}; \label{equ:xk}  
\end{align}  
(iii) the system time, i.e., the time spent by the $k$-th generated packet of source $m$ in the channel and the edge server, as  
\begin{align}  
  A_k^m = r_{i_{k}^m} - s_{i_k^m}. \label{equ:Ak_}  
\end{align}  

The Peak Age of Information (PAoI) of source $m$ at time $r_{i_k^m}$ is then calculated as  
\begin{align}  
  P_k^m = X_k^m + A_k^m. \label{equ:Pk}  
\end{align}  

Similarly, we denote $T_k^m$, $W_k^m$, and $C_k^m$ as the transmission time, waiting time, and computation time of update packet $i_k^m$, respectively. Note that in the preemptive system, $W_k^m = 0$ for all $k$. Thus, we have  
\begin{align}  
  A_k^m = T_k^m + W_k^m + C_k^m. \label{equ:Ak}  
\end{align}

Based on these definitions, Problem~\eqref{def:prob} can be reformulated as follows:  
\begin{align}  
  &\bar{P}^{*,\text{wop}} = \min_{d \in D} \sum_{m=1}^{M} w_m \lim_{N_m \to \infty} \frac{1}{N_m} \sum_{k=1}^{N_m} \mathbb{E}\left[P_{k}^m \right] \label{def:probnp} \\  
  &\bar{P}^{*,\text{wp}} = \min_{d \in D} \sum_{m=1}^{M} w_m \lim_{N_m \to \infty} \frac{1}{N_m} \sum_{n=1}^{N_m} \mathbb{E}\left[P_{k_n}^m \right] \label{def:probp},  
\end{align}  
where $N_m$ denotes the total number of packets received by the destination from source $m$, and $k_n$ is the index of the $n$-th successfully received packet from source $m$. In the non-preemptive system, $N_m$ is equal to the total number of updates sent by source $m$.  

Note that when describing the policy, the index $i$ is used to represent the $i$-th packet sent by the system, while calculating PAoI uses $m$ and $k$ to represent the $k$-th packet from source $m$.

\section{Optimal Scheduler and Sampler}

In this section, we demonstrate that for problem~\eqref{def:probnp} and problem~\eqref{def:probp}, the \textbf{random scheduler} achieves the optimal PAoI performance than any other schedulers. Then, we optimize the samplers for the two systems under the random scheduler.

\subsection{Optimal Scheduler}
For any scheduler \(\xi\), we define its scheduling frequency vector $\bm{f} =(f_1, ..., f_M)$ as
\begin{align}
  f_m=\lim_{K \to \infty} \frac{\sum_{i=1}^{K}\mathbf{1}_{m_i=m}}{K}.\label{def:f}
\end{align}

The scheduling frequency vector represents only one statistical characteristic of a scheduler. Intuitively, two schedulers with different rules should exhibit different weighted sum PAoI performance even if they have the same scheduling frequencies. However, we find that the average weighted PAoI depends solely on the scheduling frequencies. The following theorem summarizes this phenomenon.

\begin{theorem}\label{The:WOP}
 Suppose schedulers \(\xi_1\) and \(\xi_2\) have the same scheduling frequencies. For any sampler \(\pi\), there exists a sampler \(\pi'\) such that the weight sum average PAoI performance of \((\xi_2, \pi')\) is the same as \((\xi_1, \pi)\) for both problem~\eqref{def:probnp} and problem~\eqref{def:probp}.
\end{theorem}
\begin{proof}
  See Appendix~\ref{proof:The:WOP}.
\end{proof}

Theorem \ref{The:WOP} indicates that any two schedulers with the same scheduling frequencies can achieve the same weighted sum PAoI performance under appropriate samplers, even if they select source nodes based on different scheduling rules. For example, one scheduler follows a round-robin scheduling rule while the other follows a random scheduling rule. Therefore, we focus on a particular type of scheduler as follows:

\begin{definition}\label{def:RS}
  \textbf{Random Scheduler:} At each transmission moment, the scheduler randomly selects a source according to a fixed scheduling frequency vector.
\end{definition}

The advantage of random scheduler is that only a small set of parameters, i.e., the scheduling frequencies needs to be controlled, which is easier to implement compared to other schedulers. For example, while a round-robin scheduler can control the scheduling frequencies by adjusting the order in which each source is scheduled within the polling cycle, this method can become difficult when the scheduling frequencies are irregular fractions. In contrast, a random scheduler can adapt to any scheduling frequencies without increasing implementation complexity.

According to Theorem~\ref{The:WOP}, we have the following corollary.

\begin{corollary}\label{coro1}
  There exists a random scheduler and a causal sampler that is optimal for problem~\eqref{def:probnp}/\eqref{def:probp}.
\end{corollary}

Corollary~\ref{coro1} allows us to conclude that the optimal scheduler and the optimal sampler can be optimized separately, given that the optimal scheduler follows a random role. By applying the random scheduler, the optimization parameters of the new problem degrade to a set of scheduling frequencies and a causal sampler.
{
\begin{remark}
	When the distributions of transmission and computation times are independent but non-identical across sources, it can be proved that the randomized scheduler is still optimal but the scheduling frequency depends not only on the source to be scheduled, but also on the source of the previously scheduled packet. We refer it as a conditional randomized scheduler. Under such a scheduler, there are $M^2$ scheduling frequencies to be determined instead of $M$ frequencies in the i.i.d. case. Accordingly, the structure of the optimal sampler may also change. 
\end{remark}
}

% Hence, optimization Problem~\eqref{def:probnp} can be reformulated as:
% \begin{align}
%   \mathcal{P}^{*,\text{wop}}=\min_{0<f_m<1, \pi \in \Pi}\sum_{m=1}^{M}w_m\mathbb{E}\left[\lim_{N_m \to \infty} \frac{1}{N_m}\sum_{n=1}^{N_m}P_n^m \right]\label{def:probnp2}.
% \end{align}

\subsection{Optimal Sampler}

As proved in our previous single-source work \cite{jian2024opt}, the sampler's decision variables in both settings have the same upper limit and exhibit specific threshold structures. In the multi-source system, samplers with these properties still achieve the optimal PAoI performance under the random scheduler. Next, we extend the optimal samplers from the single-source system to the multi-source system.

We start with the continuous working sampler, defined as follows:
\begin{definition} \textbf{Continuous Working Sampler:}\cite[Definition 1]{jian2024opt}\label{def:CWsampler}
  A sampler is said to be a continuous working sampler if it satisfies $Z_i \le T_i + W_i + C_i$ for all $i$.
\end{definition}

{Continuous working sampler guarantees that the system is always working. In other words, the channel and the edge server will not be idle simultaneously. According to the definition, it can be seen that in Fig.~\ref{fig:Age_curve}(a), the sampling of packet 3 is continuous working, while the sampling of packet 2 is not.} Ref.~\cite{jian2024opt} shows that in a single-source system, the continuous working sampler achieves the optimal PAoI performance compared to other samplers. The following lemma demonstrates that the continuous working sampler remains the optimal sampler in multi-source systems.

\begin{Lemma}\label{lem:maxZ}
  The continuous working sampler is the optimal for both problem~\eqref{def:probnp} and problem~\eqref{def:probp}.
\end{Lemma}
\begin{proof}
See Appendix~\ref{proof:lem:maxZ}.
\end{proof}

We firstly discuss the non-preemptive systems. In the single-source non-preemptive system, the fixed threshold sampler, which is a special type of continuous working sampler, achieves the optimal PAoI performance compared to other continuous working samplers. In this sampler, the decision variable is $Z_i = T_i + W_i + \min\{C_i, \theta\}$. The following theorem demonstrates that in the multi-source non-preemptive system, the optimal sampler has the same threshold structure.

\begin{theorem}\label{The:WOP_FMT}
  Under the random scheduler, the optimal sampler for problem~\eqref{def:probnp} satisfies $Z_i=T_i+W_i+\min\{C_i, \theta^{m_{i+1}}\}$, where $\{\theta^1,...,\theta^M\}$ is a set of thresholds.
\end{theorem}
\begin{proof}
  We complete the proof in two steps: First, we show that the optimal sampler generates a time interval \(\Theta_i\) from a probability distribution based on the scheduling result $m_{i+1}$, and then determines \(Z_i\) as \(Z_i = T_i + W_i + \min\{C_i, \Theta_i \}\). Next, we prove that there exists a fixed set of thresholds that achieves the optimal PAoI performance compared to other sets of probability distributions. The details are provided in Appendix~\ref{proof:The:WOP_FMT}.
\end{proof}

{Theorem \ref{The:WOP_FMT} shows that a new packet should be transmitted a fixed source-dependent time after the previous one starts processing. The minimization operator in the theorem makes sure that the optimal sampler is continuous working, i.e., a new packet is immediately transmitted when the previous one finishes processing even if the fixed time threshold has not been exceeded.} Based on Theorem \ref{The:WOP_FMT}, we define the fixed multi-threshold sampler as follows:
\begin{definition}\textbf{Fixed Multi-Threshold Sampler:}\label{def:FMTsampler}
  A sampler is said to be a fixed multi-threshold sampler $\pi_{\text{FMT}}$ if it determines the decision variable based on a set of fixed thresholds $\bm{\theta}= (\theta^1,...,\theta^M)$ as $Z_i = T_i+W_i+\min\{C_i, \theta^{m_{i+1}}\}$.
\end{definition}

Then, we consider the preemptive system. In the single-source preemptive system, the decision variable of the optimal sampler depends on a function of the transmission time, i.e., \( Z_i = T_i + \min\{C_i, g(T_i)\} \). In the multi-source preemptive system, we observe that the sampler also depends on functions of the transmission time \( T_i \).

\begin{theorem}\label{The:WP_SD}
  Under the random scheduler, the optimal sampler for problem~\eqref{def:probp} satisfies $Z_i=T_i+\min\{C_i, g^{m_{i}}(T_i)\}$, where ${\textbf{g}}= (g^1(\cdot),...,g^M(\cdot))$ is a set of sampling functions and $g^i:\mathbb{R}^+\to \mathbb{R}^+$.
\end{theorem}
\begin{proof}
  See Appendix~\ref{proof:The:WP_SD}.
\end{proof}

Theorems \ref{The:WOP_FMT} and \ref{The:WP_SD} illustrate the distinct properties of the optimal samplers in preemptive and non-preemptive systems: when transmitting a packet, the optimal sampling action in the non-preemptive system depends on the source of the new packet, whereas in the preemptive system, the sampling action depends on the source of the last packet. The reason is that in the non-preemptive system, the sampling decision directly affects the queuing process of the current packet at the edge server. While in the preemptive system, the sampling decision influences whether the last packet is preempted.

Similarly, we define the sampler described in Theorem~\ref{The:WP_SD} as the \textit{stationary multi-transmission-aware sampler} as follows.
\begin{definition}\textbf{Stationary Multi-Transmission-aware Sampler:}
	A sampler is said to be a stationary multi-transmission-aware sampler $\pi_{\text{SMT}}$ if it determines the decision variable based on a set of sampling functions ${\textbf{g}}= (g^1(\cdot),...,g^M(\cdot))$ as $ Z_i = T_i + \min\{C_i, g^{m_{i}}(T_i)\} $.
\end{definition}

In summary, the PAoI optimization for problems \eqref{def:probnp} and \eqref{def:probp} reduces to determining a set of stationary parameters, including scheduling frequencies and sampling thresholds or functions. In the next section, we focus on algorithms to calculate these parameters.

\section{Parameter Optimizations}

In this section, we demonstrate the methods for calculating the parameters of the optimal policies for problems~\eqref{def:probnp} and \eqref{def:probp}. In the non-preemptive system, the decision parameters include the scheduling frequencies of the random scheduler and the sampling thresholds of the fixed multi-threshold sampler. In the preemptive system, the parameters include the scheduling frequencies and the sampling functions of the fixed multi-threshold sampler. In this section and the subsequent parts, we use $\xi_{R}$ to denote the random scheduler, and $\pi_{\text{FMT}}$ and $\pi_{\text{SMT}}$ to represent the fixed multi-threshold sampler and the stationary multi-transmission-aware sampler, respectively. Since the transmission and computation times are i.i.d., we use $T$ and $C$ to denote random variables following the distributions $f_T(\cdot)$ and $f_C(\cdot)$, respectively.

\subsection{The Non-preemptive System}
In the non-preemptive system, we first reformulate the objective of the problem~\eqref{def:probnp} as a function of the scheduling frequency vector $\bm f$ and the sampling thresholds $\bm \theta$. Then, we propose an algorithm to optimize it.

\subsubsection{Problem Reformulation}
The optimal policy for problem~\eqref{def:probnp} consists of a random scheduler $\xi_{\text{R}}$ and a fixed multi-threshold sampler $\pi_{\text{FMT}}$. Under the optimal policy, the decision variables $\{m_i, Z_i\}$ are stationary and ergodic, which leads to $\{X_k^m, A_k^m, P_k^m\}$ being stationary and ergodic as well. Therefore, we can omit the subscripts in problem~\eqref{def:probnp}. Then, the following corollary presents the expression for the average weighted PAoI in the optimal scheduling-sampling policy:
\begin{corollary}\label{cor:paoinp}
  In the optimal policy for problem \eqref{def:probnp}, the average weighted PAoI is given by:
  \begin{align}
    \bar{P}_{(\xi_{R},\pi_{\text{FMT}})}=&\sum_{m=1}^{M} \frac{w_{m}}{f_{m}} \sum_{n=1}^{M} f_{n} \bar Z^{n} \!+ \! \sum_{m=1}^{M} w_{m}  \bar W^{m} \! + \! \mathbb{E}\left[T\right] \!+\! \mathbb{E}\left[C\right],\label{sum_PAoI}
  \end{align}
where
\begin{align}
  \bar Z^{m} &=\mathbb{E}[T]+ \bar W^{m} +\mathbb{E}[\min\{C,\theta^m\}], \\
   \bar W^{m} &=\mathbb{E}\left[\max\{0,{C}-\theta^m-T\}\right]. \; m = 1, \cdots, M.\label{313}
\end{align}
 % The term $\bar W^{m}$ is the average waiting time of data from source $m$ and $\tilde{C}$ follows the same distribution as $C$.
\end{corollary}
\begin{proof}
  See Appendix~\ref{proof:cor:paoinp}.
\end{proof}

Based on Corollary~\ref{cor:paoinp}, problem~\eqref{def:probnp} can be transformed into the following parametric optimization problems:
\begin{subequations}
  \begin{align}
    \min_{{\bm{f}},{\theta}}\quad &\bar{P}_{(\xi_{R},\pi_{\text{FMT}})}\label{def:probnp2_obj}\\
      \text{s.t.}\quad
      &\sum_{m=1}^M f_m=1,\\
      &0<f_m<1, \forall m = 1, \cdots, M\\ 
       & \theta^m>0, \forall m = 1, \cdots, M.
  \end{align}
  \label{def:probnp2}
\end{subequations}
{The above problem is non-convex as \eqref{sum_PAoI} is a non-convex function of $\bm f$ and $\bm \theta$, which is difficult to be solved in general. In the following, we observe a closed-form relation between $\bm f$ and $\bm \theta$, and optimizing $\bm \theta$ can be decoupled into single-source threshold optimization. Based on these observations, we proposed an alternating algorithm to iteratively calculate $\bm f$ and $\bm \theta$.}

\subsubsection{Scheduling Frequency Optimization}

The following lemma shows that a solution of the scheduling frequency vector \({\bm{f}}\) for problem~\eqref{def:probnp2} can be expressed as a function of the sampling threshold vector $\bm{\theta}$:

\begin{Lemma}\label{lem:optimalf}
  Given the sampling threshold vector $\bm{\theta}$ of the fixed multi-threshold sampler, a solution of the scheduling frequency for problem~\eqref{def:probnp2} is
  \begin{equation}
    f_m=\frac{\sqrt{\frac{w_m}{\bar W^m +\mathbb{E}\left[\min \left\{C, \theta^m\right\}\right]+\mathbb{E}[T]}}}{\sum_{n=1}^M \sqrt{\frac{w_n}{\bar W^n + \mathbb{E}\left[\min \left\{C, \theta^n\right\}\right]+\mathbb{E}[T]}}}, \forall m.\label{optimalf_np}
  \end{equation}
\end{Lemma}
\begin{proof}
  See Appendix~\ref{proof:lem:optimalf}.
\end{proof}

{Lemma \ref{lem:optimalf} provides a closed-form relation between $\bm f$ and $\bm \theta$. It shows that $f_m$ is proportional to the numerator which is a function of $\bm \theta_m$. The denominator is used for normalization, which guarantees that the sum of $f_m$'s equals to 1.
}

\subsubsection{Sampling Threshold Optimization} \label{def:subsec_optimalsampling_np}
Given the scheduling frequency vector \({\bm{f}}\) and by re-organizing the summation order in \eqref{sum_PAoI} and substituting variables, the problem~\eqref{def:probnp2} can be transformed into the following problem:
\begin{align}
  \min_{{\theta}} \;&\sum_{m=1}^M \left(w_m\bar W^{m} + f_m \left[\bar W^{m} + \mathbb{E} [\min \left\{C, \theta^m\right\}]\right] \sum_{n=1}^M \frac{w_n}{f_n}  \right)\label{def:probnp_t}
\end{align}

It can be seen that the sampling thresholds $(\theta^1, \cdots,\theta^M)$ in \eqref{def:probnp_t} are independent with one another because each term containing \(\theta^m\) is additive. According to this, problem \eqref{def:probnp_t} can be decomposed into $M$ subproblems, each optimizing a parameter in the sampling threshold vector $\bm \theta$. The $m$-th subproblem is:
\begin{align}
	\min_{\theta^m}\;& w_m \bar W^{m} + f_m \left[\bar W^m + \mathbb{E}[\min \left\{C, \theta^m\right\}]\right] \sum_{n=1}^M \frac{w_n}{f_n}, \label{def:probnp_tsub}
\end{align}
where $\bar W^{m}$ is given by \eqref{313}. 

We can observe that problem \eqref{def:probnp_tsub} has the same form as the sampling threshold optimization problem in the single-source non-preemptive system\cite[Problem (10)]{jian2024opt}, with the only difference being some coefficients. Therefore, we can apply the optimal parameter $\theta^*$ calculation method in \cite[Sec. III-C]{jian2024opt} to solve problem \eqref{def:probnp_tsub}. As the $M$ subproblems are solved, problem \eqref{def:probnp_t} is also resolved.

\begin{remark}
  The method proposed in \cite{jian2024opt} determines the optimal parameter $\theta^*$ by searching for the local minima of the objective function \eqref{def:probnp_tsub}. Specifically, by analyzing the first-order and second-order conditions of the objective function \eqref{def:probnp_tsub}, its local minima can be enumerated. Then, $\theta^*$ can be determined by comparing the objective function values at these local minima. For certain special distributions, such as exponential transmission time or computation time, the local minima can be identified straightforwardly. Although this method is based on search, the obtained optimal solution is accurate and computationally efficient.  
\end{remark}

\begin{algorithm}[!t]
	\caption{Alternating Optimization Algorithm to Solve Problem \eqref{def:probnp}}\label{alg:optimal_policy}
	\begin{algorithmic}[1]
		\STATE Set arbitrary initial scheduling frequency vector ${\bm{f}}$ and sampling threshold vector $\bm{\theta}$;
		\STATE Set a sufficiently small tolerance parameter $\epsilon$, $P_{\text{old}} = 0$;
		\STATE Calculate the weighted average PAoI $P$ according to formula \eqref{sum_PAoI};
		\WHILE {$|P_{\text{old}} - P| > \epsilon$}
		\STATE Update scheduling frequency vector ${\bm{f}}$ according to equation \eqref{optimalf_np};
		\STATE Update sampling threshold vector $\bm{\theta}$ according to results in Sec.~\ref{def:subsec_optimalsampling_np};
		\STATE $P_{\text{old}} = P$.
		\STATE Update $P$ according to formula \eqref{sum_PAoI};
		\ENDWHILE
		\STATE \textbf{return} $d^*=({\bm{f}} , \bm{\theta} )$.
	\end{algorithmic}
\end{algorithm}

\subsubsection{Alternating Optimization Algorithm}

Based on the above results, we propose an alternating optimization algorithm to calculate the parameters for the scheduler and the sampler. The algorithm details are shown in Algorithm \ref{alg:optimal_policy}: The algorithm iteratively updates the scheduling frequency vector \({\bm{f}}\) and the sampling threshold \({\theta}\) until the weighted average PAoI converges. {It is remarkable that since the problem (17) is non-convex, the alternating optimization can only converge to the local optimum, which is not guaranteed to be the global optimum. However, numerical results in Sec.~\ref{sec:num} demonstrate that the proposed algorithm perform almost the same as the global optimum.}

%\begin{figure*}
%  \centering
%  \begin{align}\label{equ:barPm}
%    \bar{P}^m_{(\xi_{R},\pi_{\text{SMT}})}=\frac{\mathbb{E}[T]\frac{w_m}{f_m}+\frac{w_m}{f_m}\sum_{i=1,i\neq m}^M f_i\mathbb{E}[\min\{C,g^i(T)\}]+f_m\mathbb{E}[\min\{C,g^m(T)\}]P(\bar{\Omega}^m)\sum_{i=1}^M \frac{w_i}{f_iP(\bar{\Omega}^i)}}{P(\bar{\Omega}^m)}+\text{const}\tag{26}
%  \end{align}
%  \hrulefill
%\end{figure*}
%
%\begin{figure*}
%  \centering
%  \begin{multline}\label{equ:barPm_a}
%    F(\bm{f}, g^m)=\mathbb{E}[T]\frac{w_m}{f_m}+\frac{w_m}{f_m}\sum_{i=1,i\neq m}^M f_i\mathbb{E}[\min\{C,g^i(T)\}]+f_m\mathbb{E}[\min\{C,g^m(T)\}]P(\bar{\Omega}^m)\sum_{i=1,i\neq m}^M \frac{w_i}{f_iP(\bar{\Omega}^i)}\\+w_m\mathbb{E}[\min\{C,g^m(T)\}]\tag{27}
%  \end{multline}
%  \hrulefill
%\end{figure*}

\subsection{The Preemptive System}

The parameters in the preemptive system can be optimized by following the similar procedure.

\subsubsection{Problem Reformulation}

The optimal scheduling-sampling policy for problem~\eqref{def:probp} is a random scheduler $\xi_{\text{R}}$ and a stationary multi-transmission-aware sampler $\pi_{\text{SMT}}$, the average weighted PAoI is given by the following corollary:

\begin{corollary}\label{cor:paoip}
In the optimal policy, the average weighted PAoI in problem \eqref{def:probp} is given by:
\begin{align}
  \bar{P}_{(\xi_{R},\pi_{\text{SMT}})}\!=\!\sum_{m=1}^{M} \frac{w_m}{P(\bar{\Omega}^m)} \left\{ \frac{\mathbb{E}[Z]}{f_m}+{\mathbb{E}[(T+C)\mathbf{1}_{\bar{\Omega}^m}]}\right\},\label{sum_PAoI2}
\end{align}
where $\bar{\Omega}^m$ denotes the event of successful delivery of a packet from source $m$ and $\mathbb{E}[Z]$ represents the average inter-generation time, which are calculated as
\begin{align}
  \mathbb{E}[Z]&=\mathbb{E}[T]+\sum_{m=1}^M f_m\mathbb{E}[\min\{C,g^m(T)\}],\label{5215}\\
  P(\bar{\Omega}^m)&=P(C\le g^m(T)+\widehat{T} ),\label{51522}
\end{align}
where $\widehat{T}$ and $T$ are i.i.d. computation times.
\end{corollary}
\begin{proof}
  See Appendix \ref{proof:cor:paoip}.
\end{proof}

Based on Corollary~\ref{cor:paoip}, problem~\eqref{def:probp} can be transformed into the following parameter optimization problem:
\begin{subequations}
  \begin{align}
    \min_{{\bm{f}},\bm{g}}\quad &\bar{P}_{(\xi_{R},\pi_{\text{SMT}})}\label{def:probp2_obj}\\
      \text{s.t.}\quad &\sum_{m=1}^M f_m=1,\\
      &0<f_m<1, \forall m = 1, \cdots, M\\
      &g^m(\cdot) > 0, \forall m = 1, \cdots, M.      
  \end{align}
  \label{def:probp2}
\end{subequations}

\subsubsection{Scheduling Frequency Optimization}

Given sampling functions ${g}$, the expressions for the scheduling frequency solution of problem~\eqref{def:probnp2} are provided in the following lemma:

\begin{Lemma}\label{lem:optimalf_p}
  Given sampling functions ${g}$ of the fixed multi-threshold sampler, the scheduling frequency solution for problem~\eqref{def:probp2} is
  \begin{equation}
    f_m=\frac{\sqrt{\frac{w_m}{ P(\bar{\Omega}^m)(\mathbb{E}\left[\min \left\{C, g^m(T)\right\}\right]+\mathbb{E}[T])}}}{\sum_{n=1}^M \sqrt{\frac{w_n}{P(\bar{\Omega}^m)(\mathbb{E}\left[\min \left\{C, g^n(T)\right\}\right]+\mathbb{E}[T])}}}\label{optimalf_p}.
  \end{equation}
\end{Lemma}
\begin{proof}
  The proof follows a similar procedure with the proof of Lemma~\ref{lem:optimalf}. By utilizing Lagrangian duality theory, we set up the Lagrangian function $\mathcal{L}$ with its partial derivative $\partial \mathcal{L}/\partial f_m = 0$ and applying the KKT optimality conditions to $\mathcal{L}$, we ultimately transform and derive \eqref{optimalf_p}. Proof is completed.
\end{proof}

\subsubsection{Sampling Threshold Optimization}
Similar to \eqref{def:probnp_t}, the objective function \eqref{sum_PAoI2} is in summation form with respect to all sources. However, there is a vital difference that each summation term in \eqref{def:probnp_t} is only related to one $\theta^m$, while those in \eqref{sum_PAoI2} are correlated as \eqref{5215} is related to all $g^m$'s. Hence, the problem \eqref{def:probp2_obj} cannot be decomposed. Nevertheless, we can apply an iterative optimization algorithm to determine the sampler. In particular, we iteratively optimize a certain $g^m$ by fixing all other $g^i$'s, $\forall i \neq m$, and solve the following problem:
\begin{align}
  \min_{g^m(\cdot)}\quad & \frac{w_m}{P(\bar{\Omega}^m)} \left\{ \frac{\mathbb{E}[Z]}{f_m} + {\mathbb{E}[(T+C)\mathbf{1}_{\bar{\Omega}^m}]}\right\}. \label{def:probp2_gm}
\end{align}
Problem \eqref{def:probp2_gm} is a fractional programming problem, which can be solved iteratively using the Dinkelbach algorithm \cite{dinkelbach1967}. In particular, consider an optimization problem with a parameter $c \ge 0$:
\begin{align}
  p(c)\triangleq \min_{g^m(\cdot)}&F(g^m| {\bm{f}} , \bm{g}_m^-) - cP(\bar{\Omega}^m),\label{def:probp_tp}
\end{align}
where $\bm{g}_m^- = (g^1, \cdots, g^{(m-1)}, g^{(m+1)}, \cdots, g^M)$, and 
\begin{align}
	F(g^m| {\bm{f}} , \bm{g}_m^-) = {w_m} \left\{ \frac{\mathbb{E}[Z]}{f_m} + {\mathbb{E}[(T+C)\mathbf{1}_{\bar{\Omega}^m}]}\right\}.
\end{align}

By solving the parameterized problem \eqref{def:probp_tp} with proper value of $c$, we can find the optimal solution to problem \eqref{def:probp2_gm}. The following lemma gives the relationship between the optimal sampler for problem~\eqref{def:probp2_gm} and problem~\eqref{def:probp_tp}.

\begin{Lemma}\label{Lem:pc}
  $p(c)$ is decreasing with respect to $c$. By solving $p(c^*)=0$, we obtain the optimal solution to problem \eqref{def:probp2_gm}.
\end{Lemma}
\begin{proof}
  See Appendix \ref{proof:Lem:pc}.
\end{proof}

According to Lemma \ref{Lem:pc}, we can solve Problem \eqref{def:probp2_gm} using a two-layer iterative algorithm: 1) In the inner loop, we solve Problem \eqref{def:probp_tp} by exhaustive search for $g^m$ with given $\bm{f}$ and $\bm{g}_m^-$; 2) In the outer loop, we update $c$ such that the solution of the inner loop satisfies $p(c) = 0$. The detailed procedure is shown in Algorithm \ref{alg:Dinkelbach}.  

\begin{algorithm}[t]  
  \caption{Dinkelbach Algorithm to Obtain Optimal Sampler for Problem~\eqref{def:probp2_gm}}\label{alg:Dinkelbach}  
  \begin{algorithmic}[1]  
      \STATE $\text{Given a sufficiently small tolerance } \delta \text{ and arbitrary } c$.  
      \STATE Initialize Convergence = FALSE.  
      \WHILE {Convergence = FALSE}
      \STATE $g^m(T) = \arg\min F(g^m| {\bm{f}} , \bm{g}_m^-) - c P(\bar{\Omega}^m)$ for any $T$.  
      \STATE $p(c) = F(g^m| {\bm{f}} , \bm{g}_m^-) - c P(\bar{\Omega}^m)$.  
      \IF {$-\delta \le p(c) \le \delta$}  
      \STATE Convergence = TRUE.  
      \STATE $\textbf{return } {g^m(\cdot)}$.  
      \ELSE  
      \STATE Update $c = \frac{F(g^m| {\bm{f}} , \bm{g}_m^-)}{P(\bar{\Omega}^m)}$.  
      \ENDIF  
      \ENDWHILE  
  \end{algorithmic}  
\end{algorithm}

\subsubsection{Alternating Optimization Algorithm}

We can solve Problem \eqref{def:probp} iteratively. First, the scheduling frequency is determined according to Lemma \ref{lem:optimalf_p}. Then, for each source, the optimal sampler is derived using Algorithm \ref{alg:Dinkelbach}. By alternately updating the scheduling and the sampler, we can ultimately solve Problem \eqref{def:probp}.

\section{Numerical Results}\label{sec:num}

In this section, we compare the proposed policy with the benchmark policies in both preemptive and non-preemptive systems under different transmission and computation time distributions. We set the number of sources $M=5$, with weights $\{1/15, 2/15, 3/15, 4/15, 5/15\}$. The benchmark schedulers include the weighted Round-Robin and the Max-Age-First (MAF) scheduler \cite{bedewy2021optimal}. The benchmark sampler is the zero-wait sampler \cite{sun2017update}.

\subsection{Non-Preemptive System Performance Analysis}

Firstly, we consider the performance comparison of different strategies with exponential transmission time with parameter $\lambda$ and Gamma-distributed computation time. The average transmission time is $\mathbb{E}[T]=\frac{1}{\lambda}$. The gamma distribution has a shape parameter $k$ and a scale parameter $\beta$, and $\mathbb{E}[C]=k\beta$.

The benchmark strategies are as follows:
\begin{itemize}
  \item Weighted Round-Robin Scheduler: In each polling cycle \(K\), source \(m\) is scheduled \(w_m K\) times. An appropriate polling cycle is chosen to ensure that the number of times each source is scheduled within a cycle is an integer. For example, in this case, a feasible scheduling scheme is \(\xi = \{1, 2, 2, 3, 3, 3, 4, 4, 4, 4, 5, 5, 5, 5, 5,...\}\).
  \item MAF Scheduler: The source with the maximum AoI is always scheduled.
  \item Zero-Wait Sampler: In this sampler, new data is generated and sent when the channel is idle and the server queue is empty. In the non-preemptive system, this sampler can be viewed as a special Fixed Multi-Threshold (FMT) strategy with \(\theta^m = 0\); in the preemptive system, it is a particular Stationary Multi-Transmission-Aware (SMT) Sampler with \(g^m(\cdot) = 0\).
\end{itemize}

\begin{figure}[!t]
  \centering
  \includegraphics[width=0.65\linewidth]{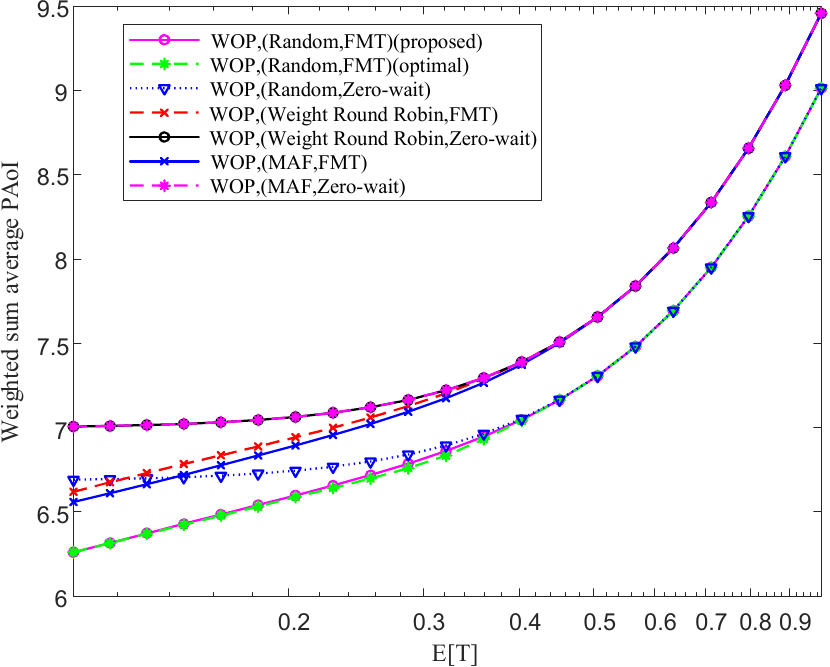}
  \caption{Average weighted PAoI under different $\mathbb{E}[T]$ with exponential transmission time and Gamma-distributed computation time in the non-preemptive system.}\label{fig:MS_TandC_exp5}
\end{figure}

\begin{figure}[!t]
  \centering
  \includegraphics[width=0.65\linewidth]{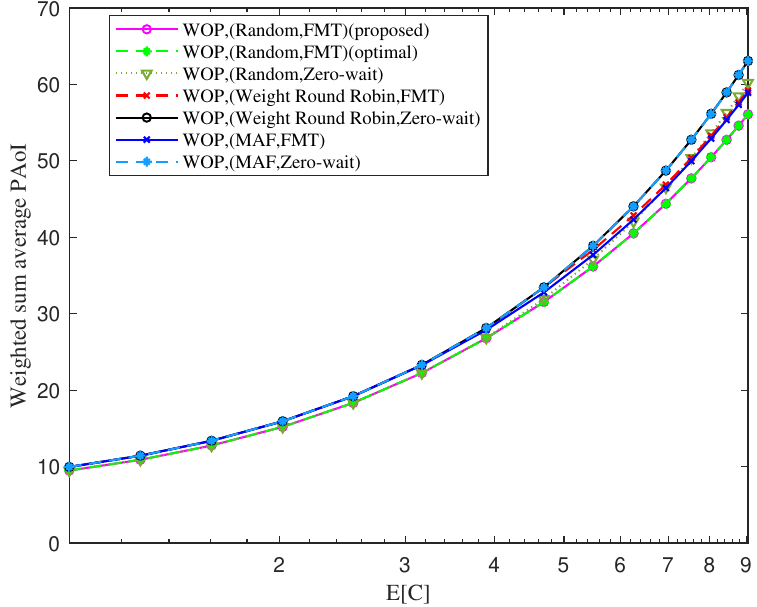}
  \caption{Average weighted PAoI under different $\mathbb{E}[C]$ with exponential transmission time and Gamma-distributed computation time in the non-preemptive system.}\label{fig:MS_TandC_exp5_1}
\end{figure}

\begin{figure}[!t]
  \centering
  \includegraphics[width=0.65\linewidth]{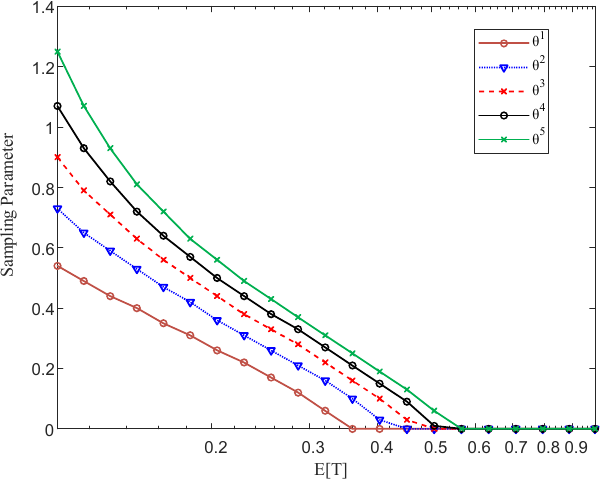}
  \caption{Optimal sampling thresholds in the FMT Sampler under different $\mathbb{E}[T]$ with exponential transmission time and Gamma-distributed computation time in the non-preemptive system.}\label{fig:MS_TandC_exp5_2}
\end{figure}

\begin{figure}[!t]
  \centering
  \includegraphics[width=0.65\linewidth]{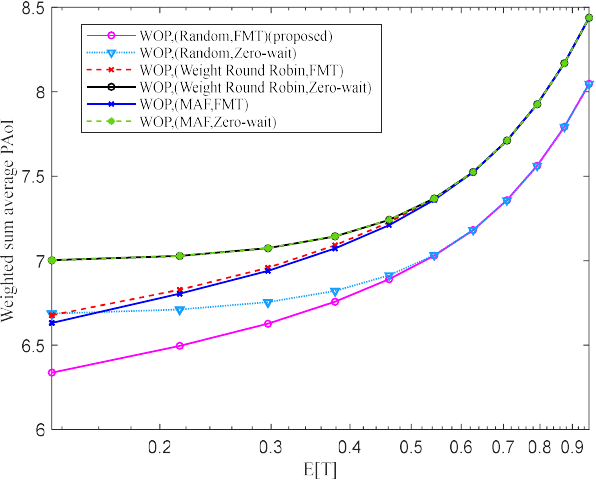}
  \caption{Average weighted PAoI under different $\mathbb{E}[T]$ with Pareto-distributed transmission time and Lognormal-distributed computation time in the non-preemptive system.}\label{fig:MS_TandC_exp5_3}
\end{figure}

{Figs.~\ref{fig:MS_TandC_exp5} and \ref{fig:MS_TandC_exp5_1} show the performance comparison between the proposed scheduler-sampler and other benchmark strategies under varying values of $\mathbb{E}[T]$ and $\mathbb{E}[C]$. In addition to the benchmarks, we also provide the PAoI performance of the optimal policy obtained via exhaustive search. Specifically, we enumerate the parameters of the FMT policy and determine the corresponding scheduler based on Lemma~\ref{lem:optimalf}. It can be observed that the proposed scheduler-sampler achieves performance very close to the optimal policy in most cases.} As \(\mathbb{E}[T]\) increases, the weighted average PAoI for all strategies grows because longer transmission times increase packet age. When the fixed sampler is the FMT Sampler, the random scheduler always outperforms the MAF Scheduler and the Weighted Round-Robin Scheduler, highlighting the impact of scheduling frequencies on PAoI performance. When the fixed scheduler is the random scheduler, the FMT Sampler performs better than the Zero-Wait Sampler, but the gap between them gradually decreases as \(\mathbb{E}[T]\) increases and becomes zero when \(\mathbb{E}[T] > 0.5\). This result can be further explained in Fig.~\ref{fig:MS_TandC_exp5_2}.

Fig.~\ref{fig:MS_TandC_exp5_2} illustrates how the sampling threshold vector \({\theta}\) in the optimal FMT Sampler within the optimal scheduler-sampler vary with \(\mathbb{E}[T]\). We observe that as \(\mathbb{E}[T]\) increases, all the thresholds gradually decrease and reduce to zero when \(\mathbb{E}[T] > 0.5\). At this point, the FMT Sampler degenerates into a Zero-Wait Sampler, which explains the results in Fig.~\ref{fig:MS_TandC_exp5}. Additionally, we can observe that as \(w_m\) increases, \(\theta^m\) also increases under fixed $\mathbb{E}[T]$. This can be explained by problem \eqref{def:probnp_tsub}: \(\theta^m\) determines the waiting time \(\mathbb{E}[W^m]\) and the variable \(\mathbb{E}[\min\{C,\theta^m\}]\), which affects the average inter-arrival time of data from source \(m\). When \(w_m\) increases, the waiting time \(\mathbb{E}[W^m]\) has a greater impact on PAoI performance, so the system tends to reduce the waiting time of data from source \(m\) in the queue, which can be achieved by increasing \(\theta^m\) according to \eqref{313}.

{We further study the performance comparison under Pareto-distributed transmission time and Lognormal-distributed computation time as shown in Fig.~\ref{fig:MS_TandC_exp5_3}. It can be seen that the performance trends are very similar to Fig.~\ref{fig:MS_TandC_exp5}. Also, our proposed policy outperforms all the other benchmarks, and converges to zero-wait sampling policy as \(\mathbb{E}[T]\) becomes large.}

\subsection{Preemptive System Performance Analysis}
\begin{figure}[!t]
  \centering
  \includegraphics[width=0.65\linewidth]{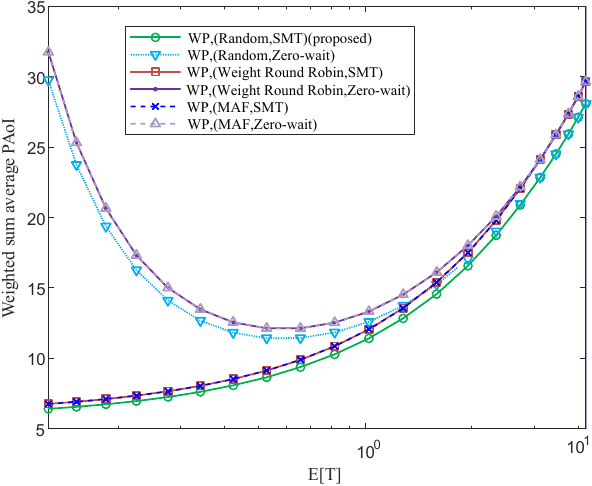}
  \caption{Average weighted PAoI under different $\mathbb{E}[T]$ with exponential transmission time and Gamma-distributed computation time in the preemptive system.}\label{fig:MS_TandC_exp6}
\end{figure}

\begin{figure}[!t]
  \centering
  \includegraphics[width=0.65\linewidth]{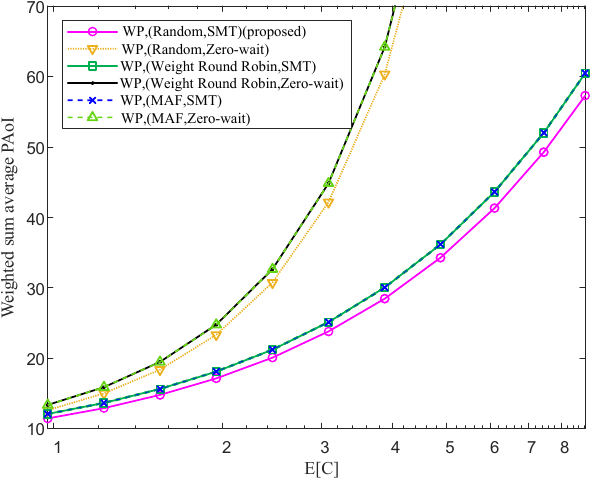}
  \caption{Average weighted PAoI under different $\mathbb{E}[C]$ with exponential transmission time and Gamma-distributed computation time in the preemptive system.}\label{fig:MS_TandC_exp6_1}
\end{figure}

\begin{figure}[!t]
  \centering
  \includegraphics[width=0.65\linewidth]{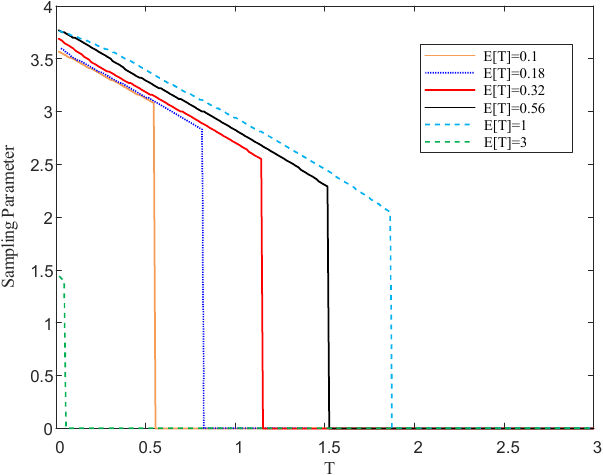}
  \caption{Optimal sampling function in the SMT Sampler under different $\mathbb{E}[T]$ with exponential transmission time and Gamma-distributed computation time in the preemptive system.}\label{fig:MS_TandC_exp6_2}
\end{figure}

\begin{figure}[!t]
  \centering
  \includegraphics[width=0.65\linewidth]{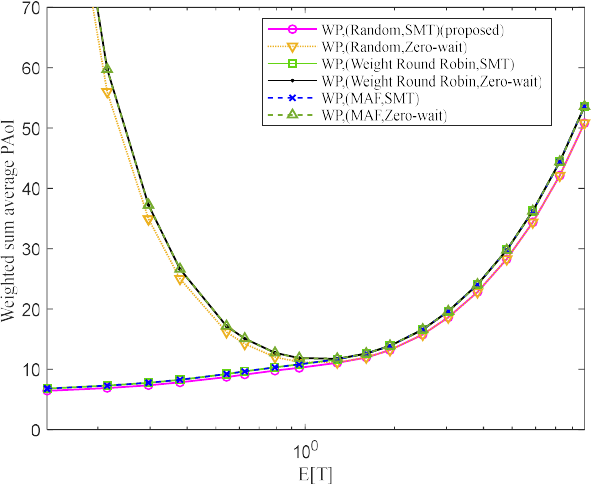}
  \caption{Average weighted PAoI under different $\mathbb{E}[T]$ with Pareto-distributed transmission time and Lognormal-distributed computation time in the preemptive system.}\label{fig:MS_TandC_exp6_3}
\end{figure}

{The performance comparison in the preemptive system with the exponential transmission time and Gamma-distributed computation time is shown in Figs.~\ref{fig:MS_TandC_exp6} and \ref{fig:MS_TandC_exp6_1}. It can be seen that, consistent with the analytical results, the random scheduler with an SMT Sampler achieves the optimal performance}. We observe that when the sampler is fixed as the Zero-Wait Sampler, the average PAoI decreases first and then increases as \(\mathbb{E}[T]\) increases. This is because, in the preemptive system, when the transmission time \(\mathbb{E}[T]\) is small, the Zero-Wait Sampler causes frequent data preemption, which increases the AoI. Therefore, the smaller the \(\mathbb{E}[T]\), the higher the average PAoI. When \(\mathbb{E}[T]\) is large, the impact of the preemption process is reduced, so the average PAoI increases as \(\mathbb{E}[T]\) increases. When the scheduler is fixed as the random scheduler, the gap between the SMT Sampler and the Zero-Wait Sampler is large when \(\mathbb{E}[T]\) is small. When \(\mathbb{E}[T]\) is large on the contrary, the performance gap vanishes. This indicates that the SMT Sampler degenerates into a Zero-Wait Sampler when \(\mathbb{E}[T]\) is large, which can be verified in Fig.~\ref{fig:MS_TandC_exp6_2}.

Fig.~\ref{fig:MS_TandC_exp6_2} illustrates how the optimal SMT sampler varies with the transmission time \(T\) under different values of \(\mathbb{E}[T]\). We use the sampling function \(g^1(T)\) as an example to demonstrate this relationship. We observe that \(g^1(T)\)  decreases to zero as \(T\) increases. This is an intuitive result: the larger \(g^1(T)\), the greater inter-generation times, reducing the likelihood of old data being preempted. Therefore, the system tends to allow data with shorter transmission time to be successfully transmitted, which results in a lower average PAoI. We observe that when \(\mathbb{E}[T] < \mathbb{E}[C] = 1\), as \(\mathbb{E}[T]\) increases, \(g^1(T)\)  gradually increases. This indicates that the preemption process is the primary factor affecting the system's PAoI performance at this stage. As $\mathbb{E}[T]$ increases, the system tends to increase \(g^1(T)\), which increases the average inter-generation time but reduces the preemption process. Since the preemption process is the main factor at this point, the average PAoI decreases as a result. We also observe that when \(\mathbb{E}[T] = 3\), \(g^1(T)\)  is non-zero only when \(T\) is very small. When \(\mathbb{E}[T]\) increases further, the SMT Sampler degenerates into a Zero-Wait Sampler, which validates the results in Fig.\ref{fig:MS_TandC_exp6} and aligns with the behavior in non-preemptive systems.

We also observe that as \(\mathbb{E}[T]\) increases, \(g^1(T)\)  gradually decreases to zero. This indicates that when \(\mathbb{E}[T]\) is large, the SMT Sampler degenerates into a Zero-Wait Sampler, which validates the results in Fig.\ref{fig:MS_TandC_exp6} and aligns with the behavior in non-preemptive systems. This demonstrates that in the preemptive system, as \(\mathbb{E}[T]\) increases, the impact of the preemption process on the average PAoI diminishes, making the increase in the average inter-arrival time due to the increase in \(\mathbb{E}[T]\) the primary factor for the increase in PAoI. Therefore, the sampler reduces \(g^1(T)\)  to minimize the average PAoI. This phenomenon can also be observed in single-source preemptive system\cite{jian2024opt}.

{Fig.~\ref{fig:MS_TandC_exp6_3} further shows the performance comparison between the proposed scheduler-sampler and other benchmark strategies under Pareto-distributed transmission time and Lognormal-distributed computation time. The results also demonstrate the similar performance trends and relations among different strategies as in Fig.~\ref{fig:MS_TandC_exp6}.}

\section{Conclusions}\label{sec:con}

In this paper, it is found that a random scheduler can achieve optimal performance in both preemption and non-preemption cases and the optimal scheduling frequencies are derived. For the non-preemptive system, we find that the optimal sampler is a fixed multi-threshold sampler. Then, an alternating optimization algorithm is proposed to compute the parameters. For the preemptive system, we prove that a stationary multi-transmission-aware sampler is optimal. Based on this, we use the Coordinate Descent and Dinkelbach algorithm to compute the sampler. In the numerical simulations, the proposed algorithms ourperform other benchmark scheduler-samplers, and perform close to the optimal. It is also observed that the optimal sampler degenerates into a Zero-Wait Sampler when the average transmission time is large.

\appendix

\subsection{Proof of Theorem~\ref{The:WOP}}\label{proof:The:WOP}

We start from problem~\eqref{def:probnp}. Since we only consider update transmission when the queue is empty, the transmitter must wait for a certain period $\Theta$ after old data starts to be processed. Furthermore, we use $\Theta_i$ to denote the time interval before the transmission of the $i$-th update, where $\Theta_i=Z_{i-1}-T_{i-1}-W_{i-1}$. Then, we have $W_{i-1}=\max\{0,Z_{i-1}-\Theta_i-T_{i-1}\}$, and use $\Theta_k^m$ to represent the time interval before the transmission of update $i_k^m$. Since $T$ and $W$ are system parameters, the sampler $\pi = \{Z_1, Z_2, ...\}$ can be redefined as $\pi = \{\Theta_1, \Theta_2, ...\}$. The scheduler $\xi = \{m_1, m_2, ...\}$ determines the source of each update. Therefore, we can represent the scheduler as $\xi=\{i_k^m|m\in\{1,...,M\},k\in\{1,2,...\}\}$. A scheduler-sampler pair can be represented as $d=(i_k^m,\Theta_k^m)$, where $m\in\{1,...,M\},k\in\{1,2,...\}$. With this definition, the weighted sum PAoI for problem~\eqref{def:probnp} is:
\begin{align}
  \bar{P}_{d}&=\sum_{m=1}^{M}w_m\lim_{N_m \to \infty} \frac{1}{N_m}\sum_{k=1}^{N_m}\mathbb{E}\left[P_k^m \right]\label{proof_The:WOP_OS_a}\\
  &=\sum_{m=1}^{M}w_m\lim_{N \to \infty} \frac{\sum_{k=1}^{f_m N}\mathbb{E}[X_k^m+A_k^m]}{f_m N} \label{proof_The:WOP_OS_b}\\
  &=\sum_{m=1}^{M}w_m\lim_{N \to \infty}\left[\frac{\sum_{i=0}^{N-1}\mathbb{E}[Z_i]}{f_m N}\!+\!\frac{\sum_{k=1}^{f_m N}\mathbb{E}[T_k^m\!+\!C_k^m]}{f_m N}\right.\nonumber\\
  &\quad\quad+\left.\frac{\sum_{k=1}^{f_m N}\mathbb{E}[\max\{0,C_{i_k^m-1}-\Theta_{k}^m-T_{i_k^m}\}]}{f_m N}\right]\label{proof_The:WOP_OS_c}\\
  &=\sum_{m=1}^{M}w_m\left[\lim_{N \to \infty} \frac{\sum_{i=0}^{N-1}\mathbb{E}[Z_i]}{f_m N}+\mathbb{E}[T]+\mathbb{E}[C]\right.\nonumber\\
  &\quad\quad +\left.\frac{\sum_{k=1}^{f_m N}\mathbb{E}[\max\{0,C_{i_k^m-1}-\Theta_{k}^m-T_{i_k^m}\}]}{f_m N}\right],\label{proof_The:WOP_OS_d}
\end{align}
where $N$ is the total number of updates sent by all sources, \eqref{proof_The:WOP_OS_b} is obtained based on \eqref{def:f}; \eqref{proof_The:WOP_OS_c} holds for two reasons. First, from \eqref{equ:zk} and \eqref{equ:xk}, we obtain  
\begin{align}  
  \textstyle \sum_{k=1}^{f_m N}X_k^m=\sum_{i=0}^{N-1}Z_i. \label{equ:2xk}  
\end{align}  
Second, by substituting $\Theta_k^m$ into $W_k^m$ and adjusting the subscripts, the relationship is further validated; \eqref{proof_The:WOP_OS_d} holds because the transmission time and the computation time from different sources follow the same distribution.

Next, we formulate a new  scheduler-sampler, where the scheduler has the same scheduling frequencies $\{f_m|m=1,...,M\}$ as the original one. The sampler maintains the same $\Theta_k^m$ as the original sampler (specifically, the sampler keeps track of the number of times each source has been scheduled, then determines the value of $\Theta_k^m$ based on the source index $i$ selected by the scheduler and the number of times source $m$ has been selected before, i.e., $k-1$). Therefore, the new scheduler-sampler can be represented as $\widetilde{d}=({\widetilde{i}}_k^m,\Theta_k^m),m\in\{1,\ldots,M\},j\in\{1,2,\ldots\}$. Under $\widetilde{d}$, the weighted sum PAoI is 
\begin{align*}
  \bar{P}_{\widetilde{d}}&=\sum_{m=1}^{M}w_m\left[\lim_{N \to \infty} \frac{\sum_{i=0}^{N-1}\mathbb{E}[Z_i]}{f_m N}+\mathbb{E}[T]+\mathbb{E}[C]\right.\\
  &\quad +\left.\frac{\sum_{k=1}^{f_m N}\mathbb{E}[\max\{0,C_{\widetilde{i}_k^m-1}-\Theta_{k}^m-T_{\widetilde{i}_k^m}\}]}{f_m N}\right]\\
  &=\sum_{m=1}^{M}w_m\left[\lim_{N \to \infty} \frac{\sum_{i=0}^{N-1}\mathbb{E}[Z_i]}{f_m N}+\mathbb{E}[T]+\mathbb{E}[C]\right.\\
  &\quad +\left.\frac{\sum_{k=1}^{f_m N}\mathbb{E}[\max\{0,C_{i_k^m-1}-\Theta_{k}^m-T_{i_k^m}\}]}{f_m N}\right] = \bar{P}_{d},
\end{align*}
where the first equation is because $\widetilde{d}$ maintains $f_m$ unchanged; the second equation is because $T$ and $C$ are both i.i.d., changing the indices does not affect the expectation. Therefore, we have successfully transformed an arbitrary causal scheduler-sampler into a random scheduler-sampler with the same performance. Thus, the proof for problem~\eqref{def:probnp} is complete.

Next we consider problem~\eqref{def:probp}. The weighted sum PAoI for problem~\eqref{def:probp} is:

\begin{align}
\bar{P}_d =&\sum_{m=1}^{M}w_m\lim_{N_m \to \infty} \frac{1}{N_m}\sum_{n=1}^{N_m}\mathbb{E}\left[P_{k_n}^m \right] \nonumber\\
=&\sum_{m=1}^{M}w_m \frac{\lim_{N_m \to \infty}\sum_{n=1}^{N_m}\mathbb{E}\left[X_{k_n}^m+T_{k_n}^m+C_{k_n}^m \right]}{\lim_{N \to \infty}\mathbb{E} \left[ \sum_{n=1}^{N} \mathbf{1}_{\Omega_n^m} \mathbf{1}_{\Omega_n} \right]}\label{equ:xa00}\\
=& \sum_{m=1}^{M} w_m \lim_{N \to \infty} \frac{\sum_{n=1}^{N} \mathbb{E} \left[ Z_{n-1} + (T_n + C_n) \mathbf{1}_{\Omega_n^m} \mathbf{1}_{\Omega_n} \right]}{\mathbb{E} \left[ \sum_{n=1}^{N} \mathbf{1}_{\Omega_n^m} \mathbf{1}_{\Omega_n} \right]} \label{equ:xa0}\\
=& \sum_{m=1}^{M} w_m  \lim_{N \to \infty}\left[ \frac{\sum_{n=1}^{N} \mathbb{E} \left[ Z_{n-1} \right]}{\sum_{k=1}^{f_m N} \mathbb{E} [\mathbf{1}_{\Omega_{i_k^m}}]} \right.  \nonumber\\
&\quad\quad\quad +\left. \frac{\sum_{k=1}^{f_m N} \mathbb{E} \left[ (T^m_k + C^m_k) \mathbf{1}_{\Omega_{i_k^m}} \right]}{\mathbb{E} \left[ \sum_{k=1}^{f_m N} \mathbf{1}_{\Omega_{i_k^m}} \right]}\right]\label{equ:xa}\\
= &\sum_{m=1}^{M} w_m  \lim_{N \to \infty} \left[\frac{\sum_{n=1}^{N} \mathbb{E} \left[ X_n \right]}{\sum_{k=1}^{f_m N} \mathbb{E} [\mathbf{1}_{\{C_{i_k^m} \le \Theta^m_k + T_{i_k^m+1}\}}]}\right.  \nonumber\\
&+ \left. \frac{\sum_{k=1}^{f_m N} \mathbb{E} \left[ (T_{i_k^m} \!+\! C_{i_k^m}) \mathbf{1}_{\{C_{i_k^m} \le \Theta^m_k + T_{i_k^m+1}\}} \right]}{\sum_{k=1}^{f_m N} \mathbb{E} [\mathbf{1}_{\{C_{i_k^m} \le \Theta^m_k \!+\! T_{i_k^m+1}\}}]}\right],\label{equ:xb}
\end{align}
where \eqref{equ:xa00} follows from \eqref{equ:Pk} and \eqref{equ:Ak}; \eqref{equ:xa0} follows from \eqref{equ:2xk}; in \eqref{equ:xa} we only consider packets from source \(m\), so the index \(k\) refers to the index of the data generated by a certain source.

Then, we construct a new scheduling-sampler pair, where the new scheduler is a random scheduler with \(\bm{f}\) as the scheduling frequency, and the new sampler maintains \(\Theta_k^m\) the same as the original scheduling-sampler according to the scheduler's selection results. Therefore, the new scheduling-sampler can be represented as \(\widetilde{d} = (\widetilde{i}_k^m, \Theta_k^m), m \in \{1, \ldots, M\}, k \in \{1, 2, \ldots\}\). Under the new scheduling-sampler, the weighted average peak age is  
\begin{align}
\bar{P}_{\widetilde{d}} &= \sum_{m=1}^{M} w_m \lim_{N \to \infty}  \left[\frac{\sum_{n=1}^{N} \mathbb{E} \left[ X_n \right]}{\sum_{k=1}^{\widetilde{f}_m N} \mathbb{E} [\mathbf{1}_{\{C_{\widetilde{i}_k^m} \le \Theta^m_k + T_{\widetilde{i}_k^m+1}\}}]} \right.\nonumber \\
&\quad +  \left.  \frac{\sum_{n=1}^{\widetilde{f}_m N} \mathbb{E} \left[ (T_{\widetilde{i}_k^m} + C_{\widetilde{i}_k^m}) \mathbf{1}_{\{C_{\widetilde{i}_k^m} \le \Theta^m_k + T_{\widetilde{i}_k^m+1}\}} \right]}{\sum_{k=1}^{\widetilde{f}_m N} \mathbb{E} [\mathbf{1}_{\{C_{\widetilde{i}_k^m} \le \Theta^m_k + T_{\widetilde{i}_k^m+1}\}}]} \right]\label{equ:xc} \\
&= \sum_{m=1}^{M} w_m  \lim_{N \to \infty}\left[ \frac{\sum_{n=1}^{N} \mathbb{E} \left[ X_n \right]}{\sum_{k=1}^{f_m N} \mathbb{E} [\mathbf{1}_{\{C_{i_k^m} \le \Theta^m_k + T_{i_k^m+1}\}}]} \right.\nonumber \\
&\quad + \left.  \frac{\sum_{n=1}^{f_m N} \mathbb{E} \left[ (T_{i_k^m} + C_{i_k^m}) \mathbf{1}_{\{C_{i_k^m} \le \Theta^m_k + T_{i_k^m+1}\}} \right]}{\sum_{k=1}^{f_m N} \mathbb{E} [\mathbf{1}_{\{C_{i_k^m} \le \Theta^m_k + T_{i_k^m+1}\}}]} \right] \label{equ:xd}\\
&= \bar{P}_d,\nonumber
\end{align}
where \eqref{equ:xc} holds because the new sampler maintains \(\Theta_k^m\) unchanged. Eq.~\eqref{equ:xd}  holds because 1) the random scheduler uses \(\bm{f}\) as the scheduling frequency, so its scheduling frequency \(\widetilde{f}_m = f_m\), and 2) since \(T\) and \(C\) are independent variables, although the index change may cause changes in their values, it does not affect the calculation of the expectation. Therefore, we successfully obtained a set of random schedulers and samplers with the same performance from any set of causal scheduling-samplers. This completes the proof.

\subsection{Proof of Lemma~\ref{lem:maxZ}}\label{proof:lem:maxZ}

We prove this lemma by contradiction. We begin with the non-preemptive case. The preemptive case can be proved in the same way. Assume that under the optimal sampler $\pi \in \Pi$, the $(i+1)$-th packet is submitted $G$ seconds after the arrival of the $i$-th packet, i.e., $s_{i+1} = r_i + G$. We then construct a new sampler $\pi'$ such that for all $j < i$, $s'_j = s_j$, and for all $j \ge i$, $s'_j = s_j - G$. In this new sampler, the $(i+1)$-th packet is submitted immediately upon the arrival of the $i$-th packet.  

The sampler $\pi'$ does not affect packets prior to the $i$-th packet; thus, for all $j < i$, $W'_j = W_j$ and $Z'_j = Z_j$. However, $\pi'$ advances the submission of the $k$-th packet and all subsequent packets. In $\pi'$, even though the $i$-th packet is submitted earlier, it is still sent when the queue is empty and the edge server is idle. Consequently, the behavior of the $i$-th packet and subsequent packets remains unchanged. In the non-preemptive case, this implies that the waiting times of the packets remain the same. In the preemptive case, it ensures that the preemption relationships among all packets are preserved. Therefore, for all $j > i$, we have $W'_j = W_j$, $Z'_j = Z_j$, and for the $i$-th packet, $W'_i = W_i$ and $Z'_i = Z_i - G$.

From \eqref{equ:Pk} and \eqref{equ:Ak}, $P_k^m$ is given as $P_k^m = X_k^m + T_k^m + W_k^m + C_k^m$. Since $W_i = W_i'$ holds for all $i$ and the scheduler remains unchanged, it follows that $(W_k^m)' = W_k^m$ for all $k$ and $m$. Furthermore, \eqref{equ:2xk} yields  
\begin{align*}
  \sum_{n=1}^{N}(X_n^{m})' &= \sum_{n=0}^{N-1}Z_n' = \sum_{n=0}^{N-1}Z_n - G = \sum_{n=1}^{N}X_n^{m} - G.  
\end{align*}  
Given $P_k^m = X_k^m + T_k^m + W_k^m + C_k^m$ from \eqref{equ:Pk} and \eqref{equ:Ak}, the weighted PAoI of the new sampler satisfies $\overline{P}_{d^{'}} < \overline{P}_{d}$. This result contradicts the assumption that $\pi$ is the optimal sampler. Therefore, the optimal policy must satisfy $s_i \le r_{i-1}$ for all $i = 1, 2, \ldots$. From \eqref{equ:dk} and \eqref{equ:zk}, the inequality $Z_i = s_{i+1} - s_i \le r_i - s_i = T_i + W_i + C_i$ is obtained.  This completes the proof.

\subsection{Proof of Theorem~\ref{The:WOP_FMT}}\label{proof:The:WOP_FMT}
We complete the proof in two steps. In the first step, we find the optimal sampler from all causal continuous working samplers. The proof is similar to that of Theorem 1 in \cite{jian2024opt}. We modify it to fit our multi-source system settings.

\textbf{Step 1:} Based on the proof in Appendix~\ref{proof:The:WOP}, we still use $\Theta_i$ to represent the time between the start of the computation of packet $i$ and the transmission of packet $i+1$. In a causal continuous working sampler, we have $Z_i = T_i + W_i + \min\{C_i, \Theta_i\}$ and $\Theta_i$ is a random value determined by all historical information before the transmission of packet $i+1$, i.e., $\{S_0, T_0, C_0, W_0, m_0, \cdots, S_{i\!-\!1}, T_{i\!-\!1}, C_{i\!-\!1}, W_{i\!-\!1}, m_{i\!-\!1}\}$, denoted as $\mathcal{H}_i$, transmission time $T_i$, waiting time $W_i$, and the scheduling results $m_i,m_{i+1}$ of packet $i$. Note that under the stochastic scheduler, the sampler is aware of the source of packet $i+1$ when transmitting it, i.e., $m_{i+1}$. Specifically, the sampler first observes $\mathcal{H}_i$, $T_i$, $W_i$, $m_i$ and $m_{i+1}$, then selects $\Theta_i$ via a conditional probability $ p_{\Theta}(\Theta_i | T_i , W_i , m_i,m_{i+1}, \mathcal{H}_i)$. According to the definition of $W_i$, we have
\begin{align}  
W_i = \max\{0, C_{i-1} - \Theta_{i-1} - T_i\}. \label{equ:Wkml}  
\end{align}  
Since $T_i$ and $C_i$ are independent of the sampler, the sampler can equivalently be written as $\pi = \{\Theta_i, i \ge 1\}$. For simplicity, we define a special value for packet $i$ as  
\begin{align}  
R_i^m \triangleq \min\{C_i, \Theta_i\} + W_{i+1}(1 + \mathbf{1}_{\Omega_{i+1}^m}), \label{def:Rkml}  
\end{align}  
where $\Omega_i^m$ is the event that packet $i$ comes from source $m$. According to \eqref{equ:Wkml}, $T_{i+1}$ also affects $R_i^m$. For packet $i$ and fixed history $\mathcal{H}_i$, we group all state update sample paths with the same $T_i, W_i, m_i$, and statistically average all these paths to get the following average $R_i^m$:  
\begin{equation}  
\widehat{R}_i(m, \gamma, \eta, l, \pmb{h}) \triangleq \mathbb{E}\left[R_i^m | T_i \!=\! \gamma, W_i \!=\!\eta,m_i\!=\!l, \mathcal{H}_i \!= \!\pmb{h}\right], \nonumber
\end{equation} 
where the expectation is over $\Theta_i, C_i$, and $T_{i+1}$.  

\begin{figure*}[!t]
  \normalsize
  \begin{align}
    &\bar{P}_d = \sum_{m=1}^{M} w_m \lim_{N_m \to \infty} \frac{1}{N_m} \sum_{k=1}^{N_m} \mathbb{E}\left[P_k^m \right] = \sum_{m=1}^{M} w_m \lim_{N \to \infty} \frac{\sum_{i=1}^{N} \mathbb{E}[Z_{i-1} + (T_i+W_i+C_i) \mathbf{1}_{\Omega_{i}^m}]}{f_m N} \label{equ:mla} \\
    &= \sum_{m=1}^{M} w_m \lim_{N \to \infty} \frac{\sum_{i=0}^{N-1} \mathbb{E}[T_i + W_i + \min\{C_i, \Theta_i\} + (T_{i+1} + W_{i+1} + C_{i+1}) \mathbf{1}_{\Omega_{i+1}^m}]}{f_m N} \label{equ:mlb} \\
    &= (1 \!+\! \sum_{m=1}^{M} \frac{w_m}{f_m}) \mathbb{E}[T] \!+\! \mathbb{E}[C] \!+\! \sum_{m=1}^{M} w_m \lim_{N \to \infty} \frac{\sum_{i=0}^{N-1} \mathbb{E}[\min\{C_i, \Theta_i\}\! +\! W_{i+1}(1 + \mathbf{1}_{\Omega_{i+1}^m})]}{f_m N}  \!+\! \lim_{N \to \infty} \frac{\mathbb{E}\left[T_0 \!+\! W_0 \!-\! T_N \!-\! W_N\right]}{N} \label{equ:mlc} \\
    &= (1 + \sum_{m=1}^{M} \frac{w_m}{f_m}) \mathbb{E}[T] + \mathbb{E}[C] + \sum_{m=1}^{M} w_m \lim_{N \to \infty} \frac{\sum_{i=0}^{N-1} \mathbb{E}[\min\{C_i, \Theta_i\} + W_{i+1}(1 + \mathbf{1}_{\Omega_{i+1}^m})]}{f_m N} \label{equ:mld} \\
    &= (1 + \sum_{m=1}^{M} \frac{w_m}{f_m}) \mathbb{E}[T] + \mathbb{E}[C] + \sum_{m=1}^{M} w_m \lim_{N \to \infty} \frac{\sum_{i=0}^{N-1} \mathbb{E}_{\mathcal{H}_i}\left[\mathbb{E}_{T_i}\left[\mathbb{E}_{m_i}\left[\mathbb{E}_{W_i}\left[\widehat{R}_i(m, T_i, W_i,m_i, \mathcal{H}_i) | W_i\right] |m_i\right]| T_i\right] | \mathcal{H}_i\right]}{f_m N}  \label{equ:mle} \\
    &\ge (1 + \sum_{m=1}^{M} \frac{w_m}{f_m}) \mathbb{E}[T] + \mathbb{E}[C] + \sum_{m=1}^{M} w_m \lim_{N \to \infty} \frac{\sum_{i=0}^{N-1} \mathbb{E}_{\mathcal{H}_i}\left[R^*(m, \mathcal{H}_i) | \mathcal{H}_i\right]}{f_m N} \label{equ:mlf} \\
    &\ge (1 + \sum_{m=1}^{M} \frac{w_m}{f_m}) \mathbb{E}[T] + \mathbb{E}[C] + \sum_{m=1}^{M} w_m \lim_{N \to \infty} \frac{\sum_{i=0}^{N-1} R_{\min}(m)}{f_m N} \label{equ:mlg} \\
    &= (1 + \sum_{m=1}^{M} \frac{w_m}{f_m}) \mathbb{E}[T] + \mathbb{E}[C] + \sum_{m=1}^{M} R_{\min}(m) \frac{w_m}{f_m}, \label{equ:mlh}
    \end{align}  
  \hrulefill
  \vspace*{2pt}
\end{figure*}

Next, we derive the lower bound of the weighted average PAoI as shown in \eqref{equ:mla}-\eqref{equ:mlh}: \eqref{equ:mla} follows from \eqref{equ:xk}, \eqref{equ:Pk} and \eqref{equ:2xk}; \eqref{equ:mlc} follows from grouping random variables with the same indices and the fact that \(T_i\) and \(C_i\) are identically distributed; \eqref{equ:mld} follows from the fact that based on \eqref{equ:Wkml}, we have \(T_i + W_i \le C_i\), and thus, $\lim_{N \to \infty} \frac{\mathbb{E}\left[T_0 + W_0 - T_N - W_N\right]}{N} = 0$. Eq.~\eqref{equ:mle} follows from \eqref{def:Rkml}; in \eqref{equ:mlf}, \(R^*(m, \mathcal{H}_i)\) represents the minimum value of \(\widehat{R}_i(m, \gamma, \eta, l, \pmb{h})\) over all possible \(T_i,W_i\) and \(m_i\); in \eqref{equ:mlg}, \(R_{\min}(m)\) represents the minimum of \(R^*(m, \mathcal{H}_i)\) over all packets and their corresponding histories, i.e., taking the minimum over all \(i\) and \(\mathcal{H}_i\).

Note that in the policy that achieves \(R_{\min}(m)\), \(\Theta_i\) is determined by a distribution \(p(\Theta|m)\triangleq p_{\Theta}(\Theta_i | T_i = \gamma, W_i = \eta, m_i=l,m_{i+1}=m, \mathcal{H}_i = \pmb{h})\), for a fixed condition, i.e., fixed values of \(T_i\), \(W_i,m_i\), and \(\mathcal{H}_i\). Note that \(T_i\) is identically distributed. Therefore, if we directly apply the distributions achieving \(R_{\min}(m)\) for all $m$ to all packets without considering other information, all inequalities in \eqref{equ:mla}-\eqref{equ:mlh} will become equalities. Under such a random policy, $\Theta_i$ depends only on $m_{i+1}$. The new sampler achieves the lower bound \eqref{equ:mlh} of all causal continuous working samplers.

\textbf{Step 2:} In a policy \( d \) consisting of a random scheduler \( \xi_{\text{R}} \) and a new sampler \( \pi \) which achieves the lower bound \eqref{equ:mlh}, we have $d$ consisting of a random scheduler and a stochastic sampler that depends solely on $m_{i+1}$. Based on \eqref{equ:Wkml} and \eqref{equ:zk}, we have $\{T_i,W_i,Z_i,m_i\}$ is stationary and ergodic, resulting in ${X_k^m,A_k^m,P_k^m}$ being stationary and ergodic as well. Consequently, we can remove the subscript \( k \) in \eqref{def:probnp} as:

\begin{align*}
  \bar{P}_{(\Theta_{R},\pi)}= \sum_{m=1}^{M} w_m \mathbb{E}\left[\lim_{N_m \to \infty} \frac{1}{N_m} \sum_{n=1}^{N_m} P_n^m \right]=\sum_{m=1}^{M} w_m\mathbb{E}\left[P^m\right].
\end{align*}

We use \(\bar{P}_{\xi_{\text{R},\text{FMT}}}(\theta^1, \ldots, \theta^M)\) to represent the average weighted sum PAoI under a random scheduler and a fixed multi-threshold sampler. Then, we obtain
\begin{align*}
\bar{P}_{(\Theta_{R}, \pi)} &= \sum_{m=1}^{M}w_m\mathbb{E}\left[P^m\right] \\
&=  \int_{0}^{\infty} p_{\Theta}^1(x_1) \ldots \int_{0}^{\infty} p_{\Theta}^M(x_M) \\
&\quad\quad\quad\bar{P}_{(\Theta_{R}, \pi_{\text{FMT}})}(\theta_1, \ldots, \theta_M) \, dx_M \ldots \, dx_1 \\
&\ge \mathop{\min}_{\theta_1, \ldots, \theta_M} \bar{P}_{(\Theta_{R}, \pi_{\text{FMT}})}(\theta_1, \ldots, \theta_M).
\end{align*}
Suppose \({\theta}^*  = \{\theta_1^*, \ldots, \theta_M^*\}\) achieves the minimum in the above inequality. Therefore, we construct a fixed multi-threshold sampler \(\pi' \in \Pi_{\text{RMT}}\) such that \(p_{\Theta}^m(\Theta' = \theta^*) = 1\), which achieves the lower bound \eqref{equ:mlh}. Hence, the performance of any causal continuous working sampler is surpassed by a fixed multi-threshold sampler. This completes the proof.

\subsection{Proof of Theorem~\ref{The:WP_SD}}\label{proof:The:WP_SD}
We prove it in two steps: first, we prove that the optimal policy for problem~\eqref{def:probp2} is stationary, and then we prove that the optimal policy is deterministic under the assumption of stationarity.

\textbf{Step 1:} Under a continuously working sampler $Z_i=T_i+\min\{C_i,\Theta_i\}$, where $\Theta_i$ is a random variable determined by all the historical information prior to packet $i$, denoted by $\{S_0,T_0,C_0,m_0,\cdots,S_{i-1},T_{i-1},C_{i-1},m_{i-1}\}$, the scheduling results $m_{i},m_{i+1}$ and $T_i$ is the transmission time of packet $i$. Specifically, the source first observes $\mathcal{H}_i$, $T_i$, $m_{i}$ and $m_{i+1}$, and then selects $\Theta_i$ through the conditional probability $p_{\Theta}(\Theta_i|T_i=\gamma,m_i=m,m_{i+1}=l,\mathcal{H}_i=\pmb{h})$. Let $\Omega_i$ denote the event that packet $i$ is successfully delivered, i.e., $\Omega_i = \{C_i\le \Theta_i+T_{i+1}\}$. Define a random value as
\begin{equation}
    R_i^m(\Theta_i)=Z_{i-1}+(T_i+C_i)\mathbf{1}_{\Omega_i}\mathbf{1}_{\Omega_i^m}.\label{def:Riml}
\end{equation}
where $\Omega_i^m$ represents the event that packet $i$ comes from source $m$. For a fixed history $\mathcal{H}_i$, we group all state update sample paths with the same $m_{i+1}$, and statistically average over all these paths to obtain the following average value for packet $i$:
\begin{equation}
    \widehat{R}_i(m,l,\pmb{h}) \triangleq \mathbb{E}\left[R_i^m(\Theta_i)|m_{i+1}=l,\mathcal{H}_i=\pmb{h}\right],\label{def:hatRml}
\end{equation}
where the expectation is taken over $\Theta_i,C_i,T_{i+1}$. Similarly, define the average value of $\mathbf{1}_{\Omega_i}\mathbf{1}_{\Omega_i^m}$ as
\begin{equation}
    \widehat{x}_i(m,l,\pmb{h}) \triangleq \mathbb{E}\left[\mathbf{1}_{\Omega_i}\mathbf{1}_{\Omega_i^m}|m_{i+1}=l,\mathcal{H}_i=\pmb{h}\right].\label{def:hatxml}
\end{equation}

\begin{figure*}[!t]
  \normalsize
    \begin{align}
      \overline{P}_{d}&=\sum_{m=1}^{M}w_m\lim_{N \to \infty} \frac{\mathbb{E}\left[\sum_{n=1}^{N}(Z_{n-1}+(T_n+C_n)\mathbf{1}_{\Omega_n}\mathbf{1}_{\Omega_n^m})\right]}{\mathbb{E}[\sum_{n=1}^{N}\mathbf{1}_{\Omega_n}\mathbf{1}_{\Omega_n^m}]}\label{CC4_0a}\\
      &=\sum_{m=1}^{M}w_m\lim_{N \to \infty} \frac{\sum_{n=1}^{N}\mathbb{E}\left[R_n^{m}(\Theta_n)\right]}{\sum_{n=1}^{N}\mathbb{E}[\mathbf{1}_{\Omega_n}\mathbf{1}_{\Omega_n^m}]}\label{CC4_a}\\
      &=\sum_{m=1}^{M}w_m\lim_{N \to \infty} \frac{\sum_{n=1}^{N}\mathbb{E}_{\mathcal{H}_n}\left[\mathbb{E}_{m_{n+1}}\left[\widehat{R}_n(m,m_{n+1},\mathcal{H}_n)|m_{n+1}\right]|\mathcal{H}_n\right]}{\sum_{n=1}^{N}\mathbb{E}[\mathbf{1}_{\Omega_n}\mathbf{1}_{\Omega_n^m}]}\label{CC4_b}\\
      &=\sum_{m=1}^{M}w_m\lim_{N \to \infty} \frac{\sum_{n=1}^{N}\mathbb{E}_{\mathcal{H}_n}\left[\mathbb{E}_{m_{n+1}}\left[\widehat{x}_n(m,m_{n+1},\mathcal{H}_n)\frac{\widehat{R}_n(m,m_{n+1},\mathcal{H}_n)}{\widehat{x}_n(m,m_{n+1},\mathcal{H}_n)}|m_{n+1}\right]|\mathcal{H}_n\right]}{\sum_{n=1}^{N}\mathbb{E}[\mathbf{1}_{\Omega_n}\mathbf{1}_{\Omega_n^m}]}\label{CC4_c}\\
      &\ge\sum_{m=1}^{M}w_m\lim_{N \to \infty} \frac{\sum_{n=1}^{N}\mathbb{E}_{\mathcal{H}_n}\left[\mathbb{E}_{m_{n+1}}\left[\widehat{x}_n(m,m_{n+1},\mathcal{H}_n)|m_{n+1}\right]R^{*}(m,\mathcal{H}_n)|\mathcal{H}_n\right]}{\sum_{n=1}^{N}\mathbb{E}[\mathbf{1}_{\Omega_n}\mathbf{1}_{\Omega_n^m}]}\label{CC4_d}\\
      &\ge\sum_{m=1}^{M}w_m\lim_{N \to \infty} \frac{\sum_{n=1}^{N}\mathbb{E}_{\mathcal{H}_n}\left[\mathbb{E}_{m_{n+1}}\left[\widehat{x}_n(m,m_{n+1},\mathcal{H}_n)|m_{n+1}\right]|\mathcal{H}_n\right]R_{\min}(m)}{\sum_{n=1}^{N}\mathbb{E}[\mathbf{1}_{\Omega_n}\mathbf{1}_{\Omega_n^m}]}\label{CC4_e}\\
      &=\sum_{m=1}^{M}w_mR_{\min}(m).\label{CC4_f}
    \end{align}
  \hrulefill
  \vspace*{2pt}
\end{figure*}

Then, we obtain the lower bound as shown in \eqref{CC4_0a}--\eqref{CC4_f}. In these inequalities, \eqref{CC4_a} and \eqref{CC4_b} follow from \eqref{def:Riml} and \eqref{def:hatRml}, respectively; in \eqref{CC4_d}, $R^*(m,\mathcal{H}_n)$ represents the minimum value of $\frac{\mathbb{E}_{m_{n+1}}\left[\widehat{R}_n(m,m_{n+1},\mathcal{H}_n)|m_{n+1}\right]}{\mathbb{E}_{m_{n+1}}\left[\widehat{x}_n(m,m_{n+1},\mathcal{H}_n)|m_{n+1}\right]}$; in \eqref{CC4_e}, $R_{\min}(m)$ denotes the minimum value of $R^*(m,\mathcal{H}_n)$ for all packets from source $m$ and the corresponding history, i.e., the minimum value across all $n$ and $\mathcal{H}_n$ from source $m$; \eqref{CC4_f} is derived from the relationship between $\widehat{x}_i$ and $\mathbf{1}_{\Omega_i}\mathbf{1}_{\Omega_i^m}$. These equations provide a lower bound for the average weighted PAoI under a continuously working sampler.

Note that in a continuously working sampler achieving $R_{\min}(m)$, $\Theta_n$ is determined by the conditional distribution $p(\Theta|m,\gamma)=p_{\Theta}(\Theta_i|T_i=\gamma,m_i=m,m_{i+1}=l,\mathcal{H}_n=\pmb{h})$ for fixed $m_{i+1}$ and \(\mathcal{H}_i\). Since $T_n$ is i.i.d. and the scheduler is random, if we directly apply the distribution achieving $R_{\min}$ to all packets without considering historical information, all inequalities in \eqref{CC4_a}-\eqref{CC4_f} become equalities. The new policy achieves the lower bound and is stationary.

Next, by converting the decision process into an MDP, it can be proven that the optimal policy is deterministic. Specifically, the sampler determines $\Theta_n$ based on the scheduler's decision $m$ and a function $g^m(T)$ of the previous data transmission time. The proof is based on the idea that for this MDP, we can find a deterministic policy with the same performance as any randomized policy. A similar proof technique can be found in our previous work \cite{jian2024opt}. This completes the proof.

\subsection{Proof of Corollary~\ref{cor:paoinp}}\label{proof:cor:paoinp}
Under the FMT sampler, the \(\theta_i\)'s are independent of each other. Thus, we have
\begin{align}
  &\lim_{N_m \to \infty}\frac{\sum_{k=1}^{N_m}\mathbb{E}[W_k^m]}{N_m}\\
  =&\lim_{N_m \to \infty} \frac{\sum_{k=1}^{N_m} \mathbb{E}[\max\{0, C_{i_k^m-1} - \theta_{i_k^m-1} - T_{i_k^m}\}]}{N_m}\\
  =&\mathbb{E}\left[\max\{0,{C}-\theta^m-T\}\right] = \bar W^m.
\end{align}
Then, we obtain
\begin{align}
  &\lim_{N \to \infty}\frac{\sum_{i=1}^{N} \mathbb{E}[Z_i]}{N}=\lim_{N \!\to \!\infty}\frac{\sum_{i=1}^{N} \mathbb{E}[T_i\!+\!W_i\!+\!\min\{C_i,\theta_i\}]}{N} \label{eq:Zi1}\\
  &=\mathbb{E}[T]+\sum_{m=1}^{M} f_m \lim_{N_m \!\to\! \infty}\frac{\sum_{k=1}^{N_m}\mathbb{E}[W_{i_k^m}+\min\{C_{i_k^m},\theta_{i_k^m}\}]}{N_m} \label{eq:Zi2}\\
  &=\mathbb{E}[T]+\sum_{m=1}^{M} f_m\left(\mathbb{E}[W^m]+\mathbb{E}[\min\{C,\theta^m\}]\right) \label{eq:Zi3}\\
  &=\sum_{m=1}^{M} f_m\left[\mathbb{E}[T]+ \mathbb{E}[W^m]+\mathbb{E}[\min\{C,\theta^m\}]\right]\\
  &=\sum_{m=1}^{M} f_m \bar Z^{m},
\end{align}
where Eqs.~\eqref{eq:Zi1}-\eqref{eq:Zi3} hold because with a random scheduler, $\theta_i = \theta^{m_{i+1}}$ follows the same distribution of $\theta^{m_{i}}$.

By substituting the above equations into \eqref{def:probnp}, we have
\begin{align}
  &\bar{P}_{(\xi_{R},\pi_{\text{FMT}})} = \sum_{m=1}^{M} w_m \left[\lim_{N_m \to \infty} \frac{\sum_{k=1}^{N_m} \mathbb{E}[P_k^m]}{N_m}\right] \label{proof:cor2_a}\\
  =&\sum_{m=1}^{M} w_m \left[\lim_{N_m \to \infty} \frac{\sum_{k=1}^{N_m} \mathbb{E}[X_k^m+T^m_k +W^m_k+C^m_k]}{N_m} \right]\label{proof:cor2_b} \\
  =&\sum_{m=1}^{M} w_m \left[\lim_{N \to \infty} \frac{\sum_{i=1}^{N} \mathbb{E}[Z_i]}{f_mN} + \mathbb{E}[T] + \mathbb{E}[C]+ \bar W^m\right] \label{proof:cor2_c}\\
  =&\sum_{m=1}^{M} w_m \left[\frac{\sum_{i=1}^{M} f_i \bar Z^{m}}{f_m} + \mathbb{E}[T] + \mathbb{E}[C] + \bar W^m\right] \label{proof:cor2_d}\\
  =&\sum_{m=1}^{M} \frac{w_{m}}{f_{m}} \sum_{i=1}^{M} f_{i} \bar Z^{m}\! +\! \sum_{m=1}^{M} w_{m} \bar W^m \!+\! \mathbb{E}\left[T\right] \!+\! \mathbb{E}\left[C\right].\label{proof:cor2_e}
\end{align}
The proof is complete.

\subsection{Proof of Corollary~\ref{cor:paoip}}\label{proof:cor:paoip}
Under the SMT sampler, the $\theta_i$'s are independent of each other. Let $\Omega_{n} = \{C_n \le \Theta_{n} + T_{n+1}\}$ denote the event that packet $n$ is delivered to the destination, and let $\Omega_{n}^m$ denote the event that packet $n$ comes from source $m$. We have  
\begin{align}
  &P(\bar{\Omega}^m) = \lim_{N \to \infty}\sum_{n=1}^N\frac{1}{N}\mathbb{E}\left[\mathbf{1}_{\Omega_{n}}|\mathbf{1}_{\Omega_{n}^m} = 1\right]\\
  =&\lim_{N \to \infty}\frac{\sum_{n=1}^{N}P(C_n\le g^m(T_n)+T_{n+1})}{N}\\
  =&P(C\le g^m(T)+\widehat{T}).
\end{align}
Then, we have
\begin{align}
  &\bar{P}_{(\xi_{R},\pi_{\text{SMT}})}= \sum_{m=1}^{M} w_m \left[\lim_{N_m \to \infty} \frac{\sum_{n=1}^{N_m} \mathbb{E}[P_{k_n}^m]}{N_m}\right] \\
  =&\sum_{m=1}^{M} w_m \left[\lim_{N_m \to \infty} \frac{\sum_{n=1}^{N_m} \mathbb{E}[X_{k_n}^m+T^m_{k_n} +C^m_{k_n}]}{N_m} \right]\\
  =&\sum_{m=1}^{M} w_m \left[ \lim_{N \to \infty}\frac{\sum_{n=1}^{N} \mathbb{E}[Z_{n-1}+(T_n+C_n)\mathbf{1}_{\Omega_{n}^m}\mathbf{1}_{\Omega_{n}}]}{\sum_{n=1}^N \mathbf{1}_{\Omega_{n}^m}\mathbf{1}_{\Omega_{n}}}\right]\\
  =& \sum_{m=1}^{M} w_m \frac{\mathbb{E}[Z]+\mathbb{E}[(T_n+C_n)\mathbf{1}_{\Omega_{n}^m}\mathbf{1}_{\Omega_{n}}]}{\mathbb{E}[\mathbf{1}_{\Omega_{n}^m}\mathbf{1}_{\Omega_{n}}]}\\
  =& \sum_{m=1}^{M} w_m \frac{\mathbb{E}[Z]+\mathbb{E}[\mathbb{E}[(T_n+C_n)\mathbf{1}_{\Omega_{n}^m}\mathbf{1}_{\Omega_{n}}|\mathbf{1}_{\Omega_{n}^m}]]} {\mathbb{E}[\mathbb{E}[\mathbf{1}_{\Omega_{n}^m}\mathbf{1}_{\Omega_{n}}|\mathbf{1}_{\Omega_{n}^m}]]}\\
  =& \sum_{m=1}^{M} w_m \frac{\mathbb{E}[Z]+f_m\mathbb{E}[(T_n+C_n)\mathbf{1}_{\bar \Omega^m}]} {f_m\mathbb{E}[\mathbf{1}_{\bar \Omega^m}]}\\
  =&\sum_{m=1}^{M} \frac{w_m}{P(\bar{\Omega}^m)} \left\{ \frac{\mathbb{E}[Z]}{f_m}+{\mathbb{E}[(T+C)\mathbf{1}_{\bar{\Omega}^m}]}\right\}
\end{align}
The proof is complete.

\subsection{Proof of Lemma~\ref{lem:optimalf}}\label{proof:lem:optimalf}
We use the Lagrange duality theory to solve problem \eqref{def:probnp2} given the sampling thresholds ${\theta}$. For any multipliers $\lambda_m \ge 0, m = 1, ..., M, \mu \ge 0$, the Lagrangian function of problem \eqref{def:probnp2} is defined as follows:
\begin{align*}
    \mathcal{L} = &\mathbb{E}[T] \sum_{m=1}^M \frac{w_m}{f_m} \!+\! \sum_{m=1}^M \frac{w_m}{f_m} \sum_{n=1}^M f_n \left(\bar W^n \!+\! \mathbb{E}\left[\min \left\{C, \theta^n\right\}\right]\right)\\
    & - \sum_{m=1}^M \lambda_m f_m + \mu\left(\sum_{m=1}^M f_m - 1\right).
\end{align*}
By applying the KKT optimality conditions to the Lagrangian function $\mathcal{L}$ and setting $\partial \mathcal{L}/\partial f_m = 0$, we obtain the KKT conditions as follows:
\begin{align}
&- \mathbb{E}[T]\frac{{w_m}}{{f_m}^2} + \left(\bar W^m + \mathbb{E}[\min \{ C, \theta^m \} ]\right) \sum_{n = 1}^M \frac{{w_n}}{{f_n}} \nonumber\\
&\quad- \frac{{w_m}}{{f_m}^2}\sum_{n = 1}^M f_n\left(\bar W^n + \mathbb{E}[\min \{ C, \theta^n \} ]\right) + \lambda_m  + \mu = 0, \label{equ:proofa}
\end{align}
and $\sum_{m=1}^M f_m = 1,\lambda_m \ge 0,f_m > 0,\lambda_m f_m = 0$.

To compute the optimal solution, we re-arrange the left and right sides of equation \eqref{equ:proofa} and multiply both sides by $f_m$ to get the following equation:
\begin{align}
  &f_m \left(\bar W^m + \mathbb{E}[\min \{ C, \theta^m \} ]\right) \sum_{n = 1}^M \frac{{w_n}}{{f_n}} + \mu f_m \nonumber\\
  =& \frac{{w_m}}{{f_m}} \sum_{n = 1}^M f_n\left(\bar W^n + \mathbb{E}[\min \{ C, \theta^n \} ]\right) + \mathbb{E}[T]\frac{{w_m}}{{f_m}}, \label{equ:prooff}
\end{align}
Note that the term $\lambda_m$ is eliminated due to $\lambda_mf_m=0$. Next, summing both sides of the equation over $m$ and performing algebraic manipulations, we obtain
\begin{equation}
  \mu = \mathbb{E}[T] \sum_{m = 1}^M \frac{{w_m}}{{f_m}}.
\end{equation}
Substituting this into \eqref{equ:prooff}, we get the following expression:
\begin{align*}
  f_m \left(\bar W^m + \mathbb{E}[\min \{ C, \theta^m \} ]\right) \sum_{n = 1}^M \frac{{w_n}}{{f_n}} + f_m \mathbb{E}[T] \sum_{n = 1}^M \frac{{w_n}}{{f_n}}& \\
  = \frac{{w_m}}{{f_m}} \sum_{n = 1}^M f_n\left(\bar W^n + \mathbb{E}[\min \{ C, \theta^n \} ]\right) + \mathbb{E}[T]\frac{{w_m}}{{f_m}}&.
\end{align*}
By denoting $A = \sum_{n = 1}^M f_n\left(\bar W^n + \mathbb{E}[\min \{ C, \theta^n \} ]\right)$ and $B = \sum_{n = 1}^M \frac{{w_n}}{{f_n}}$, we can get
\begin{align*}
	f_m=\sqrt{\frac{A w_m}{B \left(\bar W^m+ \mathbb{E}\left[\min \left\{C, \theta^m\right\}\right]+\mathbb{E}[T]\right)}}.
\end{align*}
As $A$ and $B$ are independent to the specific source index $m$, by normalizing the above equation such that $\sum_{m=1}^M f_m = 1$, we obtain the expression for the optimal solution as \eqref{optimalf_np}. The proof is complete.

\subsection{Proof of Lemma~\ref{Lem:pc}}\label{proof:Lem:pc}
Let $c_1>0$, and let the solution of problem~\eqref{def:probp_tp} be given by $g^{m}_1(\cdot)$ for $c=c_1$. Now for some $c_2>c_1$, we have
\begin{align*}
  p(c_1)&=\min_{g^m(\cdot)}&F(g^m| {\bm{f}} , \bm{g}_m^-) - c_1P(\bar{\Omega}^m)\\
  &\ge\min_{g^m(\cdot)}&F(g^m| {\bm{f}} , \bm{g}_m^-) - c_2P(\bar{\Omega}^m)\\
  &\ge p(c_2),
\end{align*}
where the last inequality follows since $g^{m}_1(\cdot)$ is also feasible in problem~\eqref{def:probp_tp} for $c=c_2$.

Next, note that both problem~\eqref{def:probp_tp} and \eqref{def:probp2_gm} have the same feasible set. In addition, if $p(c)=0$, then the objective function \eqref{def:probp2_gm} satisfies
\begin{align*}
  \frac{F(g^m| {\bm{f}} , \bm{g}_m^-)}{P(\bar{\Omega}^m)}=c.
\end{align*}
Hence, the objective function of \eqref{def:probp2_gm} is minimized by finding the minimum $c\ge 0$ such that $p(c)=0$. Finally, by the first part of lemma, there can only be one such $c$, which we denote $c^*$.

\bibliographystyle{IEEEtran}

\bibliography{IEEEabrv,refer}

\vfill

\end{document}